\pdfoutput=1
\documentclass[a4paper,UKenglish,autoref,thm-restate,numberwithinsect,final]{lipics-v2021}
\nolinenumbers

\usepackage{amsmath}

\usepackage{stmaryrd}
\SetSymbolFont{stmry}{bold}{U}{stmry}{m}{n}

\usepackage{stel-common}

\numberwithin{equation}{section}

\newcommand{\resetCurThmBraces}{%
\gdef\curThmBraceOpen{(}%
\gdef\curThmBraceClose{)}}
\resetCurThmBraces
\newcommand{\removeThmBraces}{%
\gdef\curThmBraceOpen{}%
\gdef\curThmBraceClose{}}
\resetCurThmBraces

\usepackage{etoolbox}
\patchcmd{\thmhead}{(#3)}{\curThmBraceOpen #3\curThmBraceClose }{}{}

\usepackage{bm} 

\providecommand{\catname}{\mathbf} 
\providecommand{\clsname}{\mathcal}
\providecommand{\oname}[1]{{\operatorname{\mathsf{#1}}}}

\def\defcatname#1{\expandafter\def\csname B#1\endcsname{\catname{#1}}}
\def\defcatnames#1{\ifx#1\defcatnames\else\defcatname#1\expandafter\defcatnames\fi}
\defcatnames ABCDEFGHIJKLMNOPQRSTUVWXYZ\defcatnames

\def\defclsname#1{\expandafter\def\csname C#1\endcsname{\clsname{#1}}}
\def\defclsnames#1{\ifx#1\defclsnames\else\defclsname#1\expandafter\defclsnames\fi}
\defclsnames ABCDEFGHIJKLMNOPQRSTUVWXYZ\defclsnames

\def\defbbname#1{\expandafter\def\csname BB#1\endcsname{{\bm{\mathsf{#1}}}}}
\def\defbbnames#1{\ifx#1\defbbnames\else\defbbname#1\expandafter\defbbnames\fi}
\defbbnames ABCDEFGHIJKLMNOPQRSTUVWXYZ\defbbnames

\def\Set{\catname{Set}}

\providecommand{\argument}{\operatorname{-\!-}}

\DeclareOldFontCommand{\bf}{\normalfont\bfseries}{\mathbf}
\providecommand{\mplus}{{\scriptscriptstyle\bf+}} 	       

\providecommand{\Id}{\operatorname{Id}}

\providecommand{\id}{\mathsf{id}}

\providecommand{\comp}{\mathbin{\circ}}
\providecommand{\iso}{\mathbin{\cong}}

\providecommand{\xto}[1]{\,\xrightarrow{#1}\,}

\providecommand{\dar}{\kern-1.2pt\operatorname{\downarrow}}	
\providecommand{\uar}{\kern-1.2pt\operatorname{\uparrow}}	

\providecommand{\fst}{\oname{fst}}
\providecommand{\snd}{\oname{snd}}
\providecommand{\pr}{\oname{pr}}

\providecommand{\brks}[1]{\langle #1\rangle}

\DeclareSymbolFont{Symbols}{OMS}{cmsy}{m}{n}
\DeclareMathSymbol{\iobj}{\mathord}{Symbols}{"3B}

\providecommand{\curry}{\oname{curry}}

\providecommand{\ev}{\oname{ev}}

\usepackage{stmaryrd}

\providecommand{\pacman}[1]{}					                     

\newcommand{\undefine}[1]{\let #1\relax}					                       

\providecommand{\mone}{{\text{\kern.5pt\rmfamily-}\mathsf{\kern-.5pt1}}}

\providecommand{\smin}{\smallsetminus}

\makeatletter
\def\mfix#1{\oname{#1}\@ifnextchar\bgroup\@mfix{}}	       
\def\@mfix#1{#1\@ifnextchar\bgroup\mfix{}}			           
\makeatother

\providecommand{\case}[3]{\mfix{case}{\mathbin{}#1}{of}{#2}{\kern-1pt;}{\mathbin{}#3}}

\usepackage{dsfont}
\usepackage{mathrsfs}

\renewcommand{\L}{\mathcal{L}}

\newcommand{\Sigmas}{\Sigma^{\star}}

\newcommand{\ar}{\mathsf{ar}}
\newcommand{\NT}{\mathrm{Nat}}

\newcommand{\epito}{\twoheadrightarrow}

\newcommand{\seq}{\subseteq}
\newcommand{\ol}{\overline}

\newcommand{\trc}{{\mathsf{trc}}}
\newcommand{\fn}{{\mathsf{fn}}}
\newcommand{\beh}{{\mathsf{beh}}}
\newcommand{\C}{{\mathcal{C}}}
\renewcommand{\id}{{\mathsf{id}}}
\renewcommand{\Nat}{\mathds{N}}
\renewcommand{\NT}{\mathsf{Nat}}

\newcommand{\stimes}[1]{{S\product #1}}
\newcommand{\xtimes}[1]{{X\product #1}}
\newcommand{\ytimes}[1]{{Y\product #1}}
\newcommand{\sbrks}[1]{{ \brks{\fst,#1}  }}
\newcommand{\cpair}[2]{{ #1, #2 }}
\newcommand{\tpair}[2]{{ #1 \product #2 }}

\renewcommand{\brack}[1]{\lbrack #1 \rbrack}
\newcommand{\sbrack}[1]{\llbracket #1 \rrbracket}
\newcommand{\ssbrack}[1]{\llbracket\mkern-4.7mu\llbracket #1 \rrbracket\mkern-4.7mu\rrbracket}
\newcommand{\f}{\oname{f}}

\newcommand{\takeout}[1]{\empty}

\newcommand{\ini}{\iota}
\newcommand{\ter}{\tau}
\newcommand{\wh}{\widehat}
\DeclareMathOperator{\Coalg}{\mathsf{Coalg}}
\DeclareMathOperator{\Alg}{\mathsf{Alg}}
\newcommand{\cofun}{(\overline{\rule{-1pt}{7pt}\argument})}
\renewcommand{\rho}{\varrho}

\newcommand{\opp}{\mathsf{op}}
\newcommand{\xla}[1]{\xleftarrow{~#1~}}
\newcommand{\mypowfin}{\mathscr{P}_\omega}

\newcommand{\pullbackangle}[2][]{\arrow[phantom,to path={
                     -- ($ (\tikztostart)!1cm!#2:([xshift=8cm]\tikztostart) $)
                        node[anchor=west,pos=0.0,rotate=#2,
                        inner xsep = 0]
                        {\begin{tikzpicture}[minimum
                        height=1mm,baseline=0,#1]
    \draw[-] (0,0) -- (.5em,.5em) -- (0,1em);
                        \end{tikzpicture}}}]{}}

\usepackage{ifdraft}
\ifdraft{
\usepackage[marginparwidth=2cm,nohead,nofoot,
            headsep = 1cm,
            footskip = 1cm,
            text={12.2cm,19.3cm},
            paperwidth=16.2cm,
            paperheight=23.3cm,
            marginparsep=1mm,
            marginparwidth=1.8cm,
            footskip=1cm,
            centering
            ]{geometry}
}{}

\renewcommand{\comp}{\cdot}
\renewcommand{\c}{\colon}

\usepackage[inline]{showlabels}
\showlabels{bibitem}
\usepackage{seqsplit}
\usepackage{xstring}
\usepackage{xcolor}
\usepackage{bbding}
\usepackage{microtype}



\usepackage{graphicx,csquotes}
\hypersetup{
  final=true,
  colorlinks=true,
  linkcolor=black,
  citecolor=black,
  filecolor=black,
  urlcolor=black,
  pdftitle={Stateful Structural Operational Semantics},
  pdfauthor={Sergey Goncharov, Stefan Milius, Lutz Schröder, Stelios Tsampas, Henning Urbat},
  pdfkeywords={structural operational semantics, rule formats, distributive laws},
  pdfduplex={DuplexFlipLongEdge},
  hypertexnames=false
}

\tikzset{
    commutative diagrams/.cd,
    arrow style=tikz,
    diagrams={>=stealth},
    row sep=large,
    column sep = huge
}

\ifdraft{%
  \makeatletter
  \setcounter{tocdepth}{2}
  \makeatother
}{}

\usepackage[footnote,nomargin,final]{fixme}
\FXRegisterAuthor{hu}{ahu}{HU} 
\FXRegisterAuthor{sm}{asm}{SM} 
\FXRegisterAuthor{st}{ast}{ST} 
\FXRegisterAuthor{ls}{als}{LS} 
\FXRegisterAuthor{sg}{asg}{SG} 

\usepackage{xspace}
\newcommand{\isos}{stateful SOS\xspace} 
\newcommand{\Isos}{Stateful SOS\xspace}

\theoremstyle{definition}
\newtheorem{defn}[theorem]{Definition} 
\newtheorem{rem}[theorem]{Remark} 
\newtheorem{notn}[theorem]{Notation} 

\newcommand{\xra}[1]{\xrightarrow{~#1~}}
\renewcommand{\xto}{\xra}
\let\xmpsto=\xmapsto
\renewcommand{\xmapsto}[1]{\xmpsto{~#1~}}

\newcommand{\V}{\mathcal{V}}

\title{Stateful Structural Operational Semantics}

\author{Sergey Goncharov}{Friedrich-Alexander-Universität Erlangen-Nürnberg,
  Germany}{sergey.goncharov@fau.de}{https://orcid.org/0000-0001-6924-8766}
  {Funded by the Deutsche Forschungsgemeinschaft (DFG, German Research
  Foundation) -- project number 215418801}

\author{Stefan Milius}{Friedrich-Alexander-Universität Erlangen-Nürnberg,
  Germany}{stefan.milius@fau.de}{https://orcid.org/0000-0002-2021-1644}
{Funded by the Deutsche Forschungsgemeinschaft (DFG, German Research
  Foundation) -- project number 259234802} 

\author{Lutz Schr{\"o}der}{Friedrich-Alexander-Universität
  Erlangen-Nürnberg, Germany}{lutz.schroeder@fau.de}{https://orcid.org/0000-0002-3146-5906}
{Work supported by Deutsche Forschungsgemeinschaft (DFG, German Research Foundation) as 
part of the Research and Training Group 2475 
(grant number 393541319/GRK2475/1-2019)}

\author{Stelios Tsampas}{Friedrich-Alexander-Universität Erlangen-Nürnberg,
  Germany}{stelios.tsampas@fau.de}{https://orcid.org/0000-0001-8981-2328}
{Funded by the Deutsche Forschungsgemeinschaft (DFG, German Research
  Foundation) -- project number 419850228} 

\author{Henning Urbat}{Friedrich-Alexander-Universität Erlangen-Nürnberg,
  Germany}{henning.urbat@fau.de}{https://orcid.org/0000-0002-3265-7168}
{Funded by the Deutsche Forschungsgemeinschaft (DFG, German
Research Foundation) -- project number 419850228}

\authorrunning{S.~Goncharov, S.~Milius, L.~Schr{\"o}der, S.~Tsampas, H.~Urbat} 

\Copyright{Sergey Goncharov, Stefan Milius, Lutz
  Schr{\"o}der, Stelios Tsampas, and Henning Urbat} 



\relatedversiondetails{Extended paper with full proofs}{https://arxiv.org/abs/2202.10866}

\acknowledgements{}

\hideLIPIcs  

\EventEditors{Amy P. Felty}
\EventNoEds{1}
\EventLongTitle{7th International Conference on Formal Structures for Computation and Deduction (FSCD 2022)}
\EventShortTitle{FSCD 2022}
\EventAcronym{FSCD}
\EventYear{2022}
\EventDate{August 2--5, 2022}
\EventLocation{Haifa, Israel}
\EventLogo{}
\SeriesVolume{228}
\ArticleNo{24}

\begin{document}
\let\subsectionautorefname\sectionautorefname

\maketitle

\begin{abstract}
  Compositionality of denotational semantics is an important concern
  in programming semantics. Mathematical operational semantics in the
  sense of Turi and Plotkin guarantees compositionality, but seen from
  the point of view of stateful computation it applies only to very
  fine-grained equivalences that essentially assume unrestricted
  interference by the environment between any two statements. We
  introduce the more restrictive \emph{\isos{}} rule format for
  stateful languages. We show that compositionality of two more
  coarse-grained semantics, respectively given by assuming read-only
  interference or no interference between steps, remains an
  undecidable property even for \isos{}. However, further restricting
  the rule format in a manner inspired by the \emph{cool} GSOS formats
  of Bloom and van Glabbeek, we obtain the \emph{streamlined} and
  \emph{cool} \isos{} formats, which respectively guarantee
  compositionality of the two more abstract equivalences.

\end{abstract}

\section{Introduction}
\label{sec:intro}

A key prerequisite for modular reasoning about process calculi and
programming languages is \emph{compositionality}: A denotational
semantics is compositional if the associated semantic equivalence
forms a congruence, that is, subterms of a given process or program
term may be replaced with equivalent subterms without affecting the
overall denotational meaning of the term. For instance, the classical
GSOS format of Bloom et al.~\cite{DBLP:journals/jacm/BloomIM95}
provides a unified formal representation of process languages
interpreted over non-deterministic labelled transition systems, and
guarantees that bisimilarity is compositional. Similarly, syntactic
restrictions of the GSOS format due to
Bloom~\cite{DBLP:journals/tcs/Bloom95} and van
Glabbeek~\cite{DBLP:journals/tcs/Glabbeek11} guarantee
compositionality for coarser equivalences.

More abstractly, GSOS is captured in Turi and Plotkin's bialgebraic
framework of \emph{mathematical operational
  semantics}~\cite{DBLP:conf/lics/TuriP97}, in which sets of
operational semantic rules are represented as distributive laws of a
monad over a comonad, a principle that has come to be used in widely
varying semantic settings~\cite{56f40c248cb44359beb3c28c3263838e,
  DBLP:conf/fossacs/KlinS08, DBLP:conf/lics/FioreS06,
  DBLP:journals/tcs/MiculanP16}.  In particular, Turi and Plotkin
demonstrated that GSOS rules correspond precisely to natural
transformations of type
\[
  \rho_X \colon \Sigma (X \product (\mypowfin X)^{L}) \to
  (\mypowfin\Sigmas X)^{L},
\]
where $\Sigma$ is a polynomial functor on the category of sets
(representing the signature of the process language at hand),~$L$ is a
set of (transition) labels, $\mypowfin$ is the finite power set
functor, corresponding to finitary non-determinism, and $\Sigmas$
denotes the free (term) monad on $\Sigma$. This is an instance of an
\emph{abstract GSOS law}, a natural transformation of type
$\Sigma (\Id \product T) \nat T\Sigmas$, with $T$, the \emph{behaviour
  functor}, instantiated to the functor~$\mypowfin^{L}$, which is
associated with image-finite $L$-labelled transition
systems. 

There is long-standing interest in SOS style specifications of
stateful programming
languages~\cite{DBLP:journals/jlp/Plotkin04a}. The natural
instantiation of mathematical operational semantics to this setting
would use $TX = (\stimes{(X \coproduct \term)})^{S}$ as the behaviour
functor (for a given set~$S$ of \emph{states}). This gives rise to an
extremely expressive rule format: In abstract GSOS laws of type
$\Sigma (\Id \product T) \nat T \Sigmas$, program constructs receive
their arguments as full-blown state transformers, which in particular
they can execute or probe on any number of input states. The
semantic domain provided by mathematical operational semantics in this
case is the final coalgebra
for~$T$, which consists of possibly infinite $S$-branching,
$S$-labelled trees, and thus is an instance of \emph{(coalgebraic)
  resumption semantics}~\cite{PirogGibbons15}, originally developed
for concurrent
settings~\cite{DBLP:conf/mfcs/HennessyP79,DBLP:conf/lics/Brooks93}. The
induced notion of semantic equivalence, for which the format
guarantees compositionality, is very fine-grained: Being a resumption
semantics, it assumes that programs cede complete control to the
environment between any two consecutive steps, and thus makes rather
few programs equivalent. Capturing less sceptical semantics, such as
standard sequential end-to-end net execution, in a compositional
manner has proved rather more challenging; generally speaking,
compositionality is harder for coarser equivalences because less
information is available about the behaviour of
subterms~\cite{DBLP:journals/tcs/Glabbeek11}.

In the present work, we approach this problem by restricting the rule
format to various degrees. We first note that the operational rules
typically associated to imperative languages resemble GSOS rules with
an additional input parameter, the present state. We correspondingly
introduce the \emph{\isos{}} format for the specification of stateful
languages, and show that \isos{} specifications are in an one-to-one
correspondence with natural transformations of type
\begin{displaymath}
  \delta_X\colon \stimes{\Sigma (X \product \stimes{(X \coproduct \term)})} \to
  \stimes{(\Sigmas X \coproduct \term)}.
\end{displaymath}
In a small-step operational semantics given in terms of transitions on
pairs consisting of states in~$S$ and program terms (or a termination
marker~$\checkmark\in\term$),~$\delta_X$ assigns to a given state
(in~$S$) and a program construct applied to argument variables with
given next-step operational behaviour (i.e.~an element of
$\Sigma (X \product \stimes{(X \coproduct \term)})$) its small-step
operational behaviour.
Effectively, this means that, in small-step operational semantics,
program constructs can execute and probe their arguments only on the
current state. We give a resumption semantics (over the final
coalgebra for $T$ as above) for \isos{}, and show that this semantics
agrees with the one obtained by converting~$\delta$ into a GSOS law,
in particular is compositional.

We go on to define two successive coarsenings of resumption
semantics: \emph{Trace semantics} assumes that the environment can
observe but not manipulate states reached in between successive
computation steps, and correspondingly uses the semantic domain
$(S^\mplus+S^\omega)^{S}$, the set of functions expecting an initial
state and returning a possibly terminating $S$-stream. The, yet
coarser, \emph{termination semantics} additionally abstracts from the
intermediate states of a computation, and thus is defined over the
semantic domain $(S \coproduct \term)^{S}$, the set of functions
expecting an initial state and returning either a final state or
divergence. Trace semantics has been used, e.g., in the type-theoretic
semantics of program logics~\cite{DBLP:conf/tphol/NakataU09} and in
formalizing concurrent systems that feature memory isolation
mechanisms~\cite{DBLP:journals/cl/PatrignaniC15,
  DBLP:conf/csfw/PatrignaniDP16}. Termination semantics is the
semantic domain typically associated with
\emph{big-step}~\cite{DBLP:journals/corr/abs-0808-0586,
  DBLP:conf/esop/OwensMKT16} or \emph{natural}
semantics~\cite{DBLP:conf/stacs/Kahn87}, and is a popular choice in
settings where fine architectural details are less
relevant~\cite{DBLP:journals/ita/Rutten99,
  Pitts:1999:ORF:309656.309671,
  DBLP:conf/ac/Pitts00}. \autoref{fig:sep} presents the three domains
in decreasing order of granularity and illustrates their differences
in terms of the programs they distinguish. Here, $S$ is the set of
variable stores assigning to every program variable its current
value. First, consider the programs
$\mathtt{x}\mathbin{\coloneqq} 1 \texttt{;}\ \mathtt{x}
\mathbin{\coloneqq} 2$ and $\mathtt{x} \mathbin{\coloneqq} 2$. These
are clearly equivalent in termination semantics but not in trace
semantics, as the additional initial step of the first program is visible in
trace semantics. Similarly, the programs
$\mathtt{x}\mathbin{\coloneqq} 1 \texttt{;}\ \mathtt{y}
\mathbin{\coloneqq} \mathtt{x}$ and
$\mathtt{x}\mathbin{\coloneqq} 1 \texttt{;}\ \mathtt{y}
\mathbin{\coloneqq} 1$ are clearly equivalent under trace semantics
but not under resumption semantics, as the latter assumes that the
value of~$\mathtt{x}$ may be changed by the environment between the two steps.
\begin{table}[t]
\begin{center}
\begin{tabular}{r|c|c}
  & {{($\mathtt{x}\mathbin{\coloneqq} 1 \texttt{;}\
    \mathtt{y} \mathbin{\coloneqq} \mathtt{x}$) = ($\mathtt{x}\mathbin{\coloneqq} 1 \texttt{;}\
    \mathtt{y} \mathbin{\coloneqq} 1$)}} & {($\mathtt{x}\mathbin{\coloneqq} 1 \texttt{;}\
                                    \mathtt{x} \mathbin{\coloneqq} 2$) = ($\mathtt{x} \mathbin{\coloneqq} 2$)}
   \\\hline
  $\nu\gamma.\, (\stimes{(\gamma \coproduct \term)})^{S}$ & \ding{56} & \ding{56} \\
  $(S^\mplus \coproduct S^\omega)^{S}$ & \ding{52} & \ding{56} \\
  $(S\coproduct \term)^S$ & \ding{52} & \ding{52}
\end{tabular}
\end{center}
\caption{Separating denotational domains by program equivalences.}
\label{fig:sep}
\end{table}
In fact, we show as our first main result that despite the restricted
expressiveness, it is undecidable whether the coarser program
equivalences are compositional for a given \isos{} specification. In a
subsequent step, we thus introduce two sets of syntactic restrictions
in the spirit of Bloom~\cite{DBLP:journals/tcs/Bloom95} and van
Glabbeek~\cite{DBLP:journals/tcs/Glabbeek11}, and show that these
guarantee that \isos{} specifications have compositional trace
semantics or termination semantics, respectively.

\subparagraph*{Related Work}

The above-mentioned \emph{cool} GSOS rules of
Bloom~\cite{DBLP:journals/tcs/Bloom95} and van
Glabbeek~\cite{DBLP:journals/tcs/Glabbeek11} guarantee
compositionality w.r.t.~various flavours of weak bisimilarity; they
motivate the \emph{cool \isos{}} format we introduce here. In a similar
vein, Tsampas et al.~\cite{DBLP:conf/mfcs/0001WNDP21} present abstract
compositionality criteria for weak bisimilarity in the context of
mathematical operational semantics~\cite{DBLP:conf/ctcs/Turi97}.  Weak
bisimilarity is still rather finer than the main semantics of interest
for the present work (trace semantics and termination semantics), as
it only abstracts away from steps that do not modify the state, such
as \texttt{skip}.

Abou-Saleh and
Pattinson~\cite{DBLP:journals/entcs/Abou-SalehP11,DBLP:conf/fossacs/Abou-SalehP13}
consider abstract GSOS specifications for while-languages and
construct semantics in Kleisli categories, working at a somewhat
higher level of generality than we do here, in particular
parametrizing over notions of side-effect. Roughly speaking, the
coarsest of their semantics amounts to a
\emph{steps-until-termination} semantics that counts but does not
enumerate intermediate states, and thus is coarser than trace
semantics but finer than termination semantics. They propose an abstract
\emph{condition on
  cones}~\cite[Sec.~4.4]{DBLP:journals/entcs/Abou-SalehP11} that
guarantees compositionality for steps-until-termination
semantics. This condition is hard to verify in concrete instances but
ensured by \emph{evaluation-in-context} rule
formats~\cite{DBLP:conf/fossacs/Abou-SalehP13} that correspond roughly
to our \emph{cool} \isos{} format, for which we show compositionality
even w.r.t.~termination semantics (a goal explicitly mentioned by
Abou-Saleh and
Pattinson~\cite[Section~6]{DBLP:conf/fossacs/Abou-SalehP13}). Our
\emph{streamlined} \isos{} format, which guarantees compositionality
of trace semantics, appears to be more permissive than
evaluation-in-context.

Bloom and Vandraager~\cite{bv94} and Mousavi et al.~\cite{mousavi04}
propose further SOS-style formats for computations with data and
prove compositionality results for semantic equivalences resembling
our resumption semantics. We note that these results require fairly
tedious proofs; this again highlights the advantage of the categorical
approach where they come entirely for free (see
\autoref{prop:finecong}). The \emph{Sfisl} format~\cite{mousavi04}
is shown to make trace semantics compositional, but in contrast to our
streamlined format it is not expressive enough to cover a fully
fledged while-language. Termination semantics is not considered
in either of these works.

\section{Preliminaries}\label{sec:prelim}
We assume that readers are familiar with basic notions from category theory
such as functors, natural transformations, and monads. In the
following we briefly recall some terminology concerning algebras and
coalgebras. Throughout, $\Set$ denotes the category of sets and
functions. We write $1 = \{*\}$ for the terminal object. For a pair
$X_1, X_2$ of objects we write $X_1\times X_2$ for the product with the
projections $\fst\colon X_1 \times X_2 \to X_1$ and $\snd\colon X_1\times X_2\to X_2$. For a pair of
morphisms $f_i\colon Y \to X_i$, $i = 1,2$, we let
$\brks{f_1,f_2}\colon Y \to X_1 \times X_2$ denote the unique induced
morphism. The \emph{canonical strength} of an endofunctor
$F \c \Set \to \Set$  is the natural transformation with components $\strength_{X,Y}\c \xtimes{FY}\to F(\xtimes{Y})$ defined by $\strength_{X,Y}(\cpair{x}{p}) = F(\lambda
y.\,(\cpair{x}{y}))(p)$. We usually drop the subscripts $X$ and $Y$.
%
%
%

\subparagraph*{Algebras} Given an endofunctor $F$ on a category $\C$,
an \emph{$F$-algebra} is a pair $(A,\alpha)$ of an object~$A$ (the
\emph{carrier} of the algebra) and a morphism $\alpha\colon FA\to A$
(its \emph{structure}). A \emph{homomorphism} from an $F$-algebra
$(A,\alpha)$ to an $F$-algebra $(B,\beta)$ is a morphism
$h\colon A\to B$ of~$\C$ such that $h\comp \alpha = \beta\comp
Fh$. Algebras for $F$ and their homomorphims form a category $\Alg F$, and an
\emph{initial} $F$-algebra is simply an initial object in that
category. If it exists, we denote the
initial $F$-algebra by $\mu F$ and its structure by $\ini\colon F(\mu F) \to \mu F$.

A common example of functor algebras are algebras over a signature.
An \emph{algebraic signature} consists of a set $\Sigma$ of operation
symbols together with a map $\ar\colon \Sigma\to \Nat$ associating to
every operation symbol $\f$ its \emph{arity} $\ar(\f)$. Symbols of arity $0$ are called \emph{constants}. Every signature $\Sigma$
induces the polynomial functor
$\coprod_{\f\in\Sigma} (\argument)^{\ar(\f)}$ on $\Set$, which we
denote by the same letter $\Sigma$. An algebra for the functor
$\Sigma$ then is precisely an algebra for the signature
$\Sigma$, i.e.~a set $A$ equipped with an operation
$\f_A\colon A^n \to A$ for every $n$-ary operation symbol
$\f\in \Sigma$. Homomorphisms between $\Sigma$-algebras are maps respecting the algebraic
structure.

Given a set $X$ of variables, we write $\Sigmas X$ for the
$\Sigma$-algebra of terms generated by $\Sigma$ with variables from
$X$. It is the \emph{free $\Sigma$-algebra} on $X$, that is, every map
$f\colon X\to A$ into the carrier of a $\Sigma$-algebra $(A,\alpha)$
uniquely extends to a homomorphism $\bar{f}\colon \Sigmas X\to A$. In
particular, the free algebra on the empty set is the initial algebra
$\mu \Sigma$; it is formed by all \emph{closed terms} of the
signature. As shown by Barr~\cite{Barr70}, the formation of free
algebras extends to a monad $\Sigmas\colon \Set\to \Set$, the \emph{free
monad} on $\Sigma$.
For every $\Sigma$-algebra $(A,\alpha)$ we obtain an
Eilenberg-Moore algebra $\wh\alpha\colon \Sigmas A \to A$ as the free
extension of $\id_A$. This is the map evaluating terms over $A$ in the algebra.

\subparagraph*{Coalgebras} A \emph{coalgebra} for an
endofunctor $F$ on $\C$ is a pair $(C,\gamma)$ of an object $C$ (the
\emph{carrier}) and a morphism $\gamma\colon C\to FC$ (its
\emph{structure}). A \emph{homomorphism} from an $F$-coalgebra
$(C,\gamma)$ to an $F$-coalgebra $(D,\delta)$ is a morphism
$h\colon C\to D$ such that $Fh\comp \gamma = \delta\comp h$.
Coalgebras for $F$ and their homomorphisms form a category $\Coalg F$, and a
\emph{final} coalgebra is a final object in that category. If
it exists, we denote the final $F$-coalgebra by $\nu F$ and its
structure by $\ter\colon \nu F \to F(\nu F)$, and we write
$\gamma^{\sharp}\colon (C,\gamma)\to (\nu F, \ter)$ for the unique
homomorphism.

\begin{example}\label{ex:final-coalgebras}
\begin{enumerate}
  \item Fix a set $S$. The set functor $BX=\stimes{(X+1)}$ has a final
  coalgebra carried by $\nu B = S^{\mplus}+ S^\omega$, the set of all non-empty
  possibly terminating $S$-streams. Its coalgebra structure
  $S^{\mplus}+ S^\omega \to \stimes{(S^{\mplus}+ S^\omega +1)}$ sends
  a stream $sw$ (where $s\in S$ and $w\in S^{*}+S^\omega$) to
  $(\cpair{s}{w})$ if $w\in S^{\mplus}+S^{\omega}$ and to $(\cpair{s}{*})$ if $w$ is
  empty.

\item Similarly, for the set functor $TX=(BX)^S = (\stimes{(X+1)})^S$,
  the terminal coalgebra is carried by the set of possibly infinite
  $S$-ary trees (i.e.~every node is either a leaf or has an
  $S$-indexed set of children) that have more than one node and where
  every edge is labelled by an element of $S$.
  The coalgebra structure $\nu T\to (\stimes{(\nu T+1)})^S$
  sends a tree $t$ to the map $s\mapsto (\cpair{s'}{t'})$ where $s'$ is the
  label of the edge from the root to its $s$-th child, and $t'$ is the subtree
  rooted at that child if it has more than one node, or $*$
  otherwise.
\end{enumerate}
\end{example}

\section{Stateful SOS Specifications}
\label{sec:new}

\label{sec:revisit}
We start off with an observation on the standard operational semantics for
sequential composition in imperative languages (see e.g.~Plotkin~\cite{DBLP:journals/jlp/Plotkin04a}), given by the following
rules:
\begin{equation}
  \label{eq:seq}
  \begin{array}{l@{\qquad\qquad}l}
  \inference[\texttt{seq1}]{\rets{s, p}{s\pr}}{\goes{s, (p\texttt{;}\ q)}{s\pr , q}} &
  \inference[\texttt{seq2}]{\goes{s, p}{s\pr, p\pr}}{\goes{s, (p\texttt{;}\ q)}{s\pr , (p\pr                                  \texttt{;}\ q)}}
  \end{array}
\end{equation}
Rule \texttt{seq1} asserts that if a program $p$, on input (state)~$s$,
terminates and produces a new state~$s\pr$, then the program
$p \texttt{;}\ q$, on input state $s$, evolves to program~$q$ and
produces the new state~$s\pr$. The other case is captured by rule
\texttt{seq2}, which asserts that if $p$, on input $s$, transitions to
$p\pr$ and produces~$s\pr$, then $p \texttt{;}\ q$, on input $s$,
transitions to $p\pr \texttt{;}\ q$ and produces~$s\pr$. Note
that for both rules, the \emph{input}~$s$ is the same in the premiss
and in the conclusion.
Consequently, to decide how $p \texttt{;}\ q$ transitions from~$s$ in
the next step, we need to know only how~$p$ behaves on~$s$, which we
can regard as the input of the entire rule.  This allows us to give a
concise categorical formulation of the rules \texttt{seq1} and
\texttt{seq2} in terms of a natural transformation
$\stimes{(X \product \stimes{(X \coproduct \term)})^{2}} \to
(\stimes{\Sigma^{\star}X \coproduct \term)}$ where $\Sigma$ is a signature containing the binary operation symbol~`\texttt{;}'. The transformation is defined by
\[
  (\cpair{s}{(x,\cpair{s'}{*}),(y,\_,\_\,)}) \mapsto (\cpair{s'}{y})
\qquad\text{and}\qquad
 (\cpair{s}{(x,\cpair{s'}{x'}),(y,\_,\_\,)}) \mapsto (\cpair{s'}{(x\pr \texttt{;}\ y)}).
\]
Compare the above with the interpretation obtained by instantiating
the GSOS principle~\cite{DBLP:conf/lics/TuriP97} to stateful
computations in the standard manner~\cite{DBLP:conf/ctcs/Turi97}. The
interpretation of~`\texttt{;}' is then given as a natural
transformation ${(X \product (\stimes{(X \coproduct \term)})^{S})^{2}}
\to (\stimes{(\Sigma^{\star}X \coproduct \term)})^S$ whose uncurried form
$\stimes{(X \product (\stimes{(X \coproduct \term)})^{S})^{2}}
\to \stimes{(\Sigma^{\star}X \coproduct \term)}$ is defined by
\begin{align*}
  & (\cpair{s}{(x,f),(y,\_\,)}) \mapsto \begin{cases}
    (\cpair{s'}{y}) & \mathrm{if}~f(s) = (\cpair{s'}{*}),\\
    (\cpair{s'}{(x\pr \texttt{;}\ y)}) & \mathrm{if}~f(s) = (\cpair{s\pr}{x\pr}).
  \end{cases}
\end{align*}
In this setting, the semantics of~$p \texttt{;}\ q$ receives the
next-step behaviours of~$p$,~$q$ as state transformers, and can in
principle probe these state transformers on arbitrary states (of
course, for~`\texttt{;}', this does not actually happen). By contrast,
our rule format, the \isos{} format formally introduced next, embodies
the restriction that the behaviour of a complex term on an input state
$s$ is predicated only on the behaviour of its subterms on~$s$. It is
this trade-off in expressiveness that buys our compositionality
results for \isos{} specifications.



\subparagraph*{The \Isos{} Rule Format}

We proceed to underpin the intuition given above with formal
definitions. We fix a countably infinite set
$\V = \{x_{1},x_2,\dots\}\cup \{y_{1},y_2,\dots\}$ of
\mbox{(\emph{meta-})}\emph{variables} and a countable set $S$ of
\emph{states}; in typical applications the elements of~$S$ are
variable stores. Moreover, we fix an algebraic signature $\Sigma$,
equivalently a polynomial functor also denoted~$\Sigma$
(cf.~\autoref{sec:prelim}). We think of the operations in~$\Sigma$ as
program constructs, and correspondingly, \emph{programs} are closed
$\Sigma$-terms, i.e.~terms formed using only the operations
in~$\Sigma$, with constants in~$\Sigma$ forming the base case.

\begin{defn}[Literals]
  \label{def:literals}
  A \emph{progressing $\Sigma$-literal} is an expression
  $\goes{s,p}{s',q}$ with $p, q \in \Sigmas\V$ and
  $s,s' \in S$. We say that $s$ is the \emph{input}, $p$ is the
  \emph{source}, $s'$ is the \emph{output} and $q$ is the \emph{target}
  of the literal. A \emph{terminating $\Sigma$-literal} is an
  expression $\rets{s,p}{s'}$ with $s,s' \in S$ and
  $p \in \Sigmas\V$. In this case, $s$ is the input,
  $p$ is the source and $s'$ is the output of the literal. A
  \emph{$\Sigma$-literal} (without further qualification) is
  either a progressing or a terminating $\Sigma$-literal.
\end{defn}
Our rule format shares some similarities with \emph{stream GSOS}~\cite[Def.~37]{DBLP:journals/tcs/Klin11}.

\begin{defn}[Rules]
  \label{def:rule}
  A \emph{\isos{} rule} for an $n$-ary operator $\f\in\Sigma$ is an
  expression
  \begin{equation}
    \inference{l_{1}  & \dots & l_{n}}{L}
  \end{equation}
  (or, in inline notation, $l_1\;\dots\;l_n/L$) where
  $l_1, \ldots, l_n$ (the \emph{premisses} of the rule) and $L$ (the
  \emph{conclusion} of the rule) are $\Sigma$-literals that have the same
  input $s \in S$, the \emph{input} of the rule, and satisfy the
  following conditions:
  \begin{enumerate}
  \item The source of the premiss $l_j$ is the variable $x_j$, and the target is $y_j$ if $l_j$ is progressing.
%
  \item The source of the conclusion $L$ is the term
    $\f(x_{1},\dots,x_{n})$. Moreover, if $L$ is progressing, the variables of its target
    term appear either as the source or the target of some premiss.
%

  \end{enumerate}
  The rule is \emph{progressing} if~$L$ is progressing, and otherwise
  the rule is \emph{terminating}. The \emph{trigger} of the rule is
  the tuple formed by its input $s$ together with the sequence of
  pairs
  $\overrightarrow{(s',\mathtt{c})} = (s_{1}',
  \mathtt{c}_{1}),\dots,(s_{n}',\mathtt{c}_{n})$, where $s_{j}'$ is
  the output of~$l_{j}$ and
  $\mathtt{c}_{j} \in \{\mathtt{pr},\mathtt{te}\}$ indicates
  whether~$l_{j}$ is progressing ($c_j=\mathtt{pr}$) or terminating
  ($c_j=\mathtt{te}$).
\end{defn}

%
%
\begin{defn}
  \label{def:spec}
  A \emph{\isos{} specification} is a set of \isos{} rules such that
  for each $n$-ary operator $\f$, each $s\in S$ and each sequence
  $\overrightarrow{(s',\mathtt{c})} = (s_{1}',
  \mathtt{c}_{1}),\dots,(s_{n}',\mathtt{c}_{n})$ where $s_j'\in S$ and
  $\mathtt{c}_{j} \in \{\mathtt{pr},\mathtt{te}\}$, there is
  {exactly} one rule for $\f$ with
  trigger $(s,\overrightarrow{(s',\mathtt{c})})$.
\end{defn}
\begin{notn}
  \label{rem:omit}
  By writing
  \begin{gather*}
    \inference{l_{1} & \dots & l_{j-1} & l_{j+1} & \dots & l_{n}}{L}
  \end{gather*}
  we mean the set of all stateful SOS rules of the form
  $l_{1}\; \dots\; l_{n}/L$ (with the missing premiss~$l_j$ filled in in any way
  possible). This captures the situation where the behaviour of the source $\f(x_1,\dots,x_n)$ of~$L$ does not depend on the behaviour of~$x_j$, given
  $l_1,\dots,l_{j-1}, l_{j+1}, \dots,l_{n}$.
\end{notn}


\begin{rem}\label{sec:sugar}
  The use of fixed enumerated variables $x_1,x_2,\ldots$ and
  $y_1,y_2,\ldots$ simplifies abstract reasoning about \isos{} (e.g.\
  \autoref{thm:specifications-vs-transformations} below). In
  examples, we use arbitrary variable names such as $p,q,x,y$, and we
  typically write rules using rule schemes, using hopefully
  self-explanatory notation. For instance, rule \texttt{seq1} in
  \autoref{fig:while1} (discussed in detail in \autoref{ex:while}) is
  to be understood as the set
  $\{\;\rets{s, p}{s\pr}\;/\;\goes{s, (p \texttt{;}\ q)}{s\pr , q}\mid
  s,s'\in S\}$ of stateful SOS rules, with variables~$p,q$, and rule
  \texttt{while1} as the set
  $\{\;/\rets{s, \texttt{while}~ e~p}{s}\mid s\in S,\,[e]_{s} = 0\}$
  (with premiss omitted as per \autoref{rem:omit}). Note the side condition
  $[e]_{s} = 0$ (expression~$e$ evaluates to~$0$ in state~$s$) of
  \texttt{while1}; the rule schemes and their side
  conditions need to be set up in such a way that they actually obey
  the restrictions in \autoref{def:spec}. For example, in the case of
  \texttt{while1} and \texttt{while2}, this is ensured by the
  respective side conditions ($[e]_{s} = 0$ and $[e]_{s} \neq 0$)
  being exhaustive and mutually exclusive.
\end{rem}

\begin{example}
  \label{ex:while}
  We will use a prototypical imperative language, \while{}, as a
  running example. Fix a countably infinite set $\mathcal{A}$ of
  program variables; then, the set~$S$ of \emph{stores} consists of
  all maps $s\colon\mathcal{A}\to\Nat$ whose \emph{support}
  $\{x\in \mathcal{A}\mid s(x)\neq 0\}$ is finite. We denote by
  $s_{[x \leftarrow v]}$ the result of changing the value of
  variable~$x$ to~$v$ in a store $s$. Moreover, we assume
  a set $E$ of expressions that include the arithmetic
  operations $+,-,*$, constants $n\in \Nat$ and variables $x \in \mathcal A$.  We write
  $[e]_{s}$ for the evaluation of expression $e$ under store $s$ (in the
  literature, evaluation is often defined stepwise by induction on the structure
  of the expression~\cite{DBLP:journals/jlp/Plotkin04a}; since this process does
  not affect the program state, we instead assume a denotational semantics for
  simplicity). The syntax of \while{} is given by
  the grammar
  \begin{bnf*}
    \bnfprod{prog}
    {\texttt{skip} \bnfor x \mathbin{\coloneqq} e \bnfor
      \bnfpn{prog} \texttt{;}\ \bnfpn{prog} \bnfor \texttt{while}~e~\bnfpn{prog}}
    \qquad(x \in \mathcal{A},\, e \in E),
  \end{bnf*}
  which in terms of algebraic operations means that the
  signature~$\Sigma$ includes constants \texttt{skip} and
  $x \mathbin{\coloneqq} e$ for all $x\in\mathcal{A}$, $e\in E$, a binary
  operation $\texttt{;}$ and a unary operation \texttt{while}~$e$ for
  each~$e\in E$. The corresponding polynomial functor is
  \[
    \Sigma X = \term \coproduct \mathcal{A} \product E
    \coproduct X \times X
    \coproduct E \product X.
  \]
  \begin{figure}
    \[
      \begin{array}{l@{\qquad}l}
        \inference[\texttt{skip}]{}{\rets{s, \texttt{skip}}{s}}
        &
          \inference[\texttt{asn}]{}{\rets{s, (x\mathbin{\coloneqq}
          e)}{s_{[x \leftarrow
          [e]_{s}]}}} \\[4ex]
        \inference[\texttt{while1}]{}{\rets{s, \texttt{while}~
        e~p}{s}}{[e]_{s} = 0}
        &
          \inference[\texttt{while2}]{}{\goes{s, \texttt{while}~e~p}{s ,(p \texttt{;}\
          \texttt{while}~e~p)}}{[e]_{s} \neq 0} \\[4ex]
        \inference[\texttt{seq1}]{\rets{s, p}{s\pr}}{\goes{s, (p \texttt{;}\
        q)}{s\pr , q}}
        &
          \inference[\texttt{seq2}]{\goes{s, p}{s\pr, p\pr}}{\goes{s, (p
          \texttt{;}\ q)}{s\pr , (p\pr\texttt{;}\ q)}}
      \end{array}
    \]
    \caption{Operational semantics of \while{}.}
    \label{fig:while1}
  \end{figure}%
  The operational semantics of \while{} in the form of a \isos{}
  specification is shown in \autoref{fig:while1}, using rule schemes as
  per \autoref{sec:sugar}.
\end{example}
\noindent
As indicated by the discussion at the beginning of this section, \isos{} specifications can be represented as natural transformations:
\begin{defn}\label{def:stateful-sos-law}
  A \emph{\isos{}
    law} is a natural transformation
  \[
    \delta_X\c\stimes{\Sigma (X \product \stimes{(X \coproduct \term)})}
    \to \stimes{(\Sigmas X\coproduct \term)}\qquad (X\in \Set).
  \]
\end{defn}

\begin{rem}\label{rem:spec-vs-trafo}

\begin{enumerate}
\item\label{rem:spec-to-trafo} Every \isos{} specification $\L$ yields a \isos{} law
\[
    \delta_X= [\delta^{\f}_X]_{\f\in \Sigma}\c\stimes{\Sigma (X \product \stimes{(X \coproduct \term)})}
    \to \stimes{(\Sigmas X \coproduct \term)} \qquad (X\in \Set)
  \]
  by distributing $\stimes{(-)}$ over
  $\Sigma (X \product \stimes{(X \coproduct \term)})$ and copairing
  the maps
\begin{equation}\label{eq:deltaf}
  \delta_{X}^{\f} \c \stimes{(X \times \stimes{(X \coproduct \term)})^{\ar(\f)}} \to \stimes{(\Sigmas X \coproduct \term)}\qquad (\f\in \Sigma)
\end{equation}
defined as follows. Given $(\cpair{s}{((v_1,\cpair{s'_1}{w_1}),\dots,(v_n,\cpair{s'_n}{w_n}))})\in \stimes{(X \times \stimes{(X \coproduct \term)})^{n}}$ with $n=\ar(f)$, let $l_1\;\dots\;l_n/L$ be the unique rule in $\L$ with source $\f$ and trigger $(s,((s'_1,c_1),\dots,(s'_n,c_n))$ where
$c_j=\mathtt{pr}$ if $w_j\in X$ and $c_j=\mathtt{te}$ if $w_j=*$. Let $s'$ be the output of $L$. Then
$\delta_X^{\f}(\cpair{s}{((v_1,\cpair{s'_1}{w_1}),\dots,(v_n,\cpair{s'_n}{w_n}))})$ is $(s',*)$ if the rule is terminating, and otherwise $(s',t')$ where $t'\in \Sigmas X$ is the term obtained from the target $t\in \Sigmas \V$ of $L$ by substituting $x_j$ by $v_j$ and $y_j$ by $w_j$ (the latter whenever $c_j=\mathtt{pr}$).
\item\label{rem:trafo-to-spec} Conversely, every \isos{} law $\delta$ yields a \isos{} specification $\L$ whose rules are defined as follows. For every $n$-ary operation symbol~$\f\in \Sigma$, $s,s_1',\ldots, s_n'\in S$ and $W \seq \{1,\ldots,n\}$, let $(\cpair{s'}{t})$ be the value of $\delta_\V^{\f}$  on $(\cpair{s}{((x_1,\cpair{s'_1}{w_1}),\dots,(x_n,\cpair{s'_n}{w_n}))})$  where $w_j=y_j$ if $j\in W$ and $w_j=*$ otherwise. If $t\in \Sigmas \V$, then $\L$ contains the rule
    \begin{displaymath}
      \inference{(s,x_j\to s_j',y_j)_{j\in W}\qquad(\rets{s,x_j}{s_j'})_{j\in\{1,\ldots,n\}\smin W}}{s,\f(x_1,\ldots,x_{n})\to s',t}, 
     \quad
    \end{displaymath}
and if $t=*$, then $\L$ contains the rule
    \begin{displaymath}
      \inference{(s,x_j\to s_j',y_j)_{j\in W}\qquad(\rets{s,x_j}{s_j'})_{j\in\{1,\ldots,n\}\smin W}}{\rets{s,\f(x_1,\ldots,x_{n})}{s'}}. 
    \end{displaymath}
\end{enumerate}
\end{rem}

\begin{theorem}\label{thm:specifications-vs-transformations}
  There is a bijective correspondence between (1) \isos{}
  specifications, (2) \isos{} laws, and (3) families of
  maps of the form
  \[
    \big(r_{\f,W}\c \stimes{S^{\ar(\f)}}\to \stimes{\Sigmas(\ar(\f)+W)}+
    S\big)_{\f\in\Sigma,W\subseteq \ar(\f)}.
  \]
  Here we identify the natural number $\ar(\f)$ with the set
  $\{1,\ldots,\ar(\f)\}$.
\end{theorem}
The correspondence between~(1) and~(2) is given by the translations of
\autoref{rem:spec-vs-trafo}, and the correspondence between~(2)
and~(3) is shown using the Yoneda lemma.

\section{Categorical Semantics and Compositionality}
\label{sec:abstract}

We proceed to develop a categorical treatment of \isos{} along the
lines of mathematical operational semantics in the style of Turi and
Plotkin~\cite{DBLP:conf/lics/TuriP97} and
Bartels~\cite{56f40c248cb44359beb3c28c3263838e}. Furthermore, we shall define
two semantic domains of interest, both coarser than the one initially
obtained through Turi-Plotkin semantics, and show that the problem of
whether a given \isos{} specification is compositional is
undecidable. We recall that if the denotational semantics of a
programming language is given by a map
$\sbrack{-}\colon \mu\Sigma \to D$ into a semantic domain $D$, then it
is called \emph{compositional} if the corresponding behavioural equivalence forms a
congruence, that is, for every $n$-ary operator $\f\in \Sigma$ and
programs $p_i,q_i\in \mu\Sigma$ ($i=1,\ldots,n$),
\[
  \text{$\sbrack{p_i}=\sbrack{q_i}$ for $i=1,\ldots,n$}
  \qquad\text{implies}\qquad
  \sbrack{\f(p_1,\ldots,p_n)}=\sbrack{\f(q_1,\ldots,q_n)}.
\]
Compositionality asserts that subprograms of a program~$p$ may be
replaced with equivalent subprograms without affecting the semantics
of~$p$, and thus allows modular reasoning.

\subsection{GSOS Laws}
\label{sec:gsos}

Turi and Plotkin's \emph{mathematical operational
  semantics}~\cite{DBLP:conf/lics/TuriP97} identifies sets of rules in
structural operational semantics (SOS) with distributive laws of
various types on a cartesian base category. We will work more
specifically with distributive laws of free monads over cofree
copointed functors on the base category $\Set$, where the free monad
is associated to a polynomial functor.  Such distributive laws can
equivalently be presented as follows.
\begin{defn}
  \label{def:gsos}
  Given a polynomial functor $\Sigma$ and an endofunctor $T$ on $\Set$, a \emph{GSOS law of}
  $\Sigma$ \emph{over} $T$ is a natural transformation
  $\rho \c \Sigma (\Id \product T) \nat T\Sigmas$.
\end{defn}
%
%
%

\noindent We shall see below that \isos{} laws determine
GSOS laws. The interested reader may find further examples of GSOS
laws in the
literature~\cite{DBLP:conf/ctcs/Turi97,56f40c248cb44359beb3c28c3263838e,DBLP:journals/tcs/Klin11}. Roughly
speaking, the input of~$\rho$ is a (program) operation applied to pairs
each consisting of a meta-variable and its assumed next-step behaviour
(encapsulated in~$T$), and the output is a next-step behaviour
reaching poststates given as programs with meta-variables.

Given a GSOS law $\rho$, the initial $\Sigma$-algebra can be equipped with a unique $T$-coalgebra structure
$\gamma \c {\mu\Sigma \to T(\mu\Sigma)}$ such that the diagram
\begin{equation}\label{eq:gamma}
  \begin{tikzcd}
    \Sigma (\mu\Sigma)
    \arrow[rr, "\ini"]
    \ar{d}[swap]{\Sigma \langle \id, \gamma \rangle}
    & &
    \mu\Sigma
    \arrow[d, dashed, "\gamma"]
    \\
    \Sigma (\mu\Sigma \product T(\mu\Sigma))
    \arrow[r, "\rho_{\mu\Sigma}"]
    & T \Sigmas (\mu\Sigma)
    \arrow[r, "T \hat\ini"]
    &
    T (\mu\Sigma)
  \end{tikzcd}
\end{equation}
commutes (see \autoref{sec:prelim} for the notation). The coalgebra
$(\mu\Sigma,\gamma)$ is called the \emph{operational model} of
$\rho$.
Dually, assuming the existence of a final coalgebra $\nu T$, there is
a unique $\Sigma$-algebra structure
$\alpha \c \Sigma(\nu T) \to \nu T$ such that the following diagram
commutes:
\begin{equation}\label{eq:alpha}
  \begin{tikzcd}[column sep = 30]
    \Sigma (\nu T)
    \ar{r}{\Sigma \langle \id, \ter \rangle}
    \ar[dashed]{d}[swap]{\alpha}
    &
    \Sigma (\nu T \product T(\nu T))
    \arrow[r, "\rho_{\nu T}"]
    &
    T
    \Sigmas (\nu T)
    \arrow[d, "T \hat \alpha"]
    \\
    \nu T
    \arrow[rr, "\ter"]
    & &
    T (\nu T)
  \end{tikzcd}
\end{equation}
The algebra $(\nu T, \alpha)$ is the \emph{denotational
  model} of $\rho$. A fundamental well-behavedness property of GSOS laws is that the
unique $\Sigma$-algebra homomorphism $(\mu \Sigma, \ini) \to (\nu T,
\alpha)$ and the unique $T$-coalgebra homomorphism $(\mu\Sigma,\gamma) \to
(\nu T, \ter)$ coincide. We denote this morphism by
\begin{equation}\label{eq:beh}
  \beh_{\rho} \c \mu\Sigma \to \nu T,
\end{equation}
and we think of it as assigning to programs their denotational
behaviour. Compositionality of this semantics is immediate from the
fact that $\beh_\rho$ is a $\Sigma$-algebra homomorphism.

\subsection{Semantic Domains for \Isos{}}
We proceed to introduce three denotational semantics of \isos{}, in
order of increasing abstraction:  \emph{resumption semantics}, in
which the program essentially cedes control to the environment between
any two program steps; \emph{trace semantics}, where the environment
may observe but not manipulate the state between program steps; and
\emph{termination semantics}, in which only the effect of executing
the program end-to-end is observable.

\begin{notn}
From now on, we instantiate the functor $T$ of \autoref{def:gsos} to 
\[ TX = (\stimes{(X \coproduct \term)})^{S},\]
for a fixed set $S$ of states. Thus $T$ represents state
transformers with possible non-termination. 
\end{notn}

\subparagraph{Resumption semantics} Every \isos{}
law~$\delta$ (see \autoref{def:stateful-sos-law})
canonically induces a GSOS law
\[\hat{\delta} \c \Sigma(\Id \product T) \nat T\Sigmas.\]
 This will guarantee
compositionality for the most fine-grained of our semantics, which we
shall refer to as \emph{resumption semantics}, via established methods
of mathematical operational semantics as recalled above. Details are
as follows. The component~$\hat\delta_X$ is obtained by currying the composite
\begin{equation}\label{eq:embed}
\begin{aligned}
  \stimes{\Sigma(X \product TX)}
  &
  \xra{\hspace*{3.7mm}\sbrks{\strength}\hspace*{3.7mm}}
  \stimes{\Sigma(\stimes{(X \product TX)})}
  \iso
  \stimes{\Sigma(X \product (\stimes{TX}))}
  \\
  &
  \xra{\tpair{\id}{\Sigma(\id \product \ev )}}
  \stimes{\Sigma (X \product \stimes{(X \coproduct\term)})}
  \xra{\delta_X}
  \stimes{(\Sigmas X \coproduct \term)},
  \end{aligned}
\end{equation}
where
$\strength\colon \stimes{\Sigma(X\times TX)} \to \Sigma(\stimes{(X \times TX)})$ is the strength (cf.~\autoref{sec:prelim}) and
$\ev \colon S \times TX = S \times (\stimes{(X+1)})^S \to \stimes{(X+1)}$ denotes the evaluation map.  Recall from
\autoref{ex:final-coalgebras} that the final coalgebra for $T$ is
carried by the set of possibly infinite $S$-branching trees, with edges labelled in $S$. Using \eqref{eq:gamma} we obtain the operational model
$\gamma \c \mu\Sigma \to T (\mu\Sigma)$ associated to $\hat\delta$. In terms of \isos{} specifications, it can be described as follows.

\begin{defn}
  \label{def:transition}
  Given a \isos{} specification $\mathcal{L}$, its \emph{transition
    function} is the map
  \[\gamma_0 \c \stimes{\mu\Sigma} \to \stimes{(\mu\Sigma \coproduct
  \term)}\] inductively defined by
  \begin{align*}
    \gamma_0(\cpair{s}{\f(t_{1},\dots,t_{n})}) & = m(\delta^{\f}_{\mu\Sigma}(\cpair{s}{(d_{1},\dots,d_{n})}))
  \end{align*}
  where 
\[d_{j} = (\cpair{t_j}{\gamma_0(\cpair{s}{t_{j}})})\qquad\text{and}\qquad
  m = \big( \stimes{(\Sigmas(\mu\Sigma) \coproduct \term)}
  \xra{\tpair{\id}{(\hat\ini \coproduct \id)}} \stimes{(\mu\Sigma \coproduct
  \term)}\big),
\]
  using the term evaluation map
  $\hat\ini\colon \Sigmas(\mu \Sigma) \to \mu\Sigma$, and $\delta_{\mu\Sigma}^{\f}$ as in \eqref{eq:deltaf}. Thus, $\gamma_0(s,p)$ performs the first computation step of program $p$ on input $s$ according to the specification $\L$.
 We
  write
  \[
    \goes{s,p}{s\pr,p\pr}
    \qquad\text{and}\qquad
    \rets{s,p}{s\pr}
  \]
  if $\gamma_0(\cpair{s}{p}) = (\cpair{s'}{p'})$ and $\gamma_0(\cpair{s}{p}) = (\cpair{s'}{*})$, respectively.
\end{defn}

\begin{proposition}\label{P:curry-h}
  Let $\L$ be a \isos{} specification with its associated transition function~$\gamma_0$
  and operational model $\gamma$. Then
  \[
    \gamma=\oname{curry}(\gamma_0)\colon \mu\Sigma \to (\stimes{(\mu\Sigma + 1)})^S.
  \]
\end{proposition}
The proof makes use of an induction principle
  that combines
  \emph{primitive recursion} (see
  e.g.~\cite[Prop.~2.4.7]{DBLP:books/cu/J2016}) and \emph{induction
    with parameters} (see
  e.g.~\cite[Exercise~2.5.5]{DBLP:books/cu/J2016}).

\begin{defn}
  The \emph{resumption semantics} of a \isos{} specification $\L$ is
  given by
  \[
    \brack{-}_{\L}=\beh_{\hat{\delta}} \colon \mu\Sigma\to \nu T,
  \]
  where $\delta$ is the \isos{} law associated to $\L$,
  $\hat\delta$ is as per \eqref{eq:embed}, and $\beh$ is defined in
  \eqref{eq:beh}. Let $\sim_\L$ denote the corresponding behavioural
  equivalence, that is, $p\sim_\L q$ iff
  $\brack{p}_\L = \brack{q}_\L$ for a given pair $p,q\in \mu\Sigma$.  We drop
  subscripts if $\L$ is clear from the context.
\end{defn}

\noindent Note that since $T$ preserves weak pullbacks, $\sim_\L$
coincides with $T$-bisimilarity in the operational model
$\gamma\colon\mu\Sigma\to T(\mu\Sigma)$~\cite{rutten00}. From the discussion in \autoref{sec:gsos} we immediately get
\begin{theorem}
  \label{prop:finecong}
  The resumption semantics of \isos{} specifications is compositional.
\end{theorem}
\noindent Resumption semantics is very fine-grained, essentially
because it does not pass the output state of a computation step on as
the input state of the next step; that is, resumption semantics
assumes that the environment takes complete control in between
steps. For instance, consider the \while{} programs%
\[
  t_{1}
  =
  \big(\mathtt{x} \mathbin{\coloneqq} 1\texttt{;}\ \mathtt{x}\mathbin{\coloneqq} \mathtt{x} + 1\big)
  \qquad\text{and}\qquad
  t_{2}
  =
  \big(\mathtt{x} \mathbin{\coloneqq} 1\texttt{;}\ \mathtt{x} \mathbin{\coloneqq}
  \mathtt{x} * 2\big).
\]
The resumption semantics of these programs in each case consists in an
$S$-branching tree of depth~$2$, in which the edge from the root to its $s$-th child is labelled $s[x\leftarrow 1]$ and the edges at the next level
are correspondingly labelled according to the effect of the
assignments $\mathtt{x}\mathbin{\coloneqq} \mathtt{x} + 1$ and
$\mathtt{x} \mathbin{\coloneqq} \mathtt{x} * 2$, respectively. In
particular, the semantics of the two programs differ -- as intuitively
expected under a resumption semantics, since the environment may
manipulate the value of~$x$ in between the two assignments. To obtain
a more coarse-grained notion of process equivalence, we have to
quotient the semantic domain $\nu T$ further.




\subparagraph{Trace Semantics}
\label{sec:quot1}
Consider the set functor $B$ given by
\[BX = \stimes{(X \coproduct \term)};\]
thus $TX=(BX)^S$. Recall from \autoref{ex:final-coalgebras} that the final
coalgebra $\nu B$ is carried by the set $S^{\mplus}+S^{\omega}$ of
possibly terminating $S$-streams. The set~$(\nu B)^{S}$ serves as the
semantic domain for \emph{trace semantics} for imperative
programs~\cite{DBLP:conf/tphol/NakataU09,
  DBLP:journals/cl/PatrignaniC15, DBLP:conf/csfw/PatrignaniDP16},
which associates to a program the possibly terminating sequence of
states it computes from a given initial state. In order to formally
introduce trace semantics in our setting, we proceed to construct a
quotient map $\nu T \epito (\nu B)^{S}$ by coinduction. To this
end, we define the functor $\cofun \c \Coalg T \to \Coalg B$, which
maps a $T$-coalgebra $(C,\zeta)$ to the $B$-coalgebra
\[
  \bar{\zeta}
  =
  \stimes{C}
  \xrightarrow{\tpair{\id}{\zeta}}
  \stimes{(BC)^S}
  \xto{\ev}
  BC = \stimes{(C+1)}
  \xto{\sbrks{\strength}}
  \stimes{(\stimes{C} + 1)}
  =
  B(\stimes{C}),
\]
where $\strength\colon \stimes{(C+1)} \to \stimes{C} + 1$ is the
strength of the functor $(\argument) + 1$, given by
$(\cpair{s}{c}) \mapsto (\cpair{s}{c})$ and
$(\cpair{s}{*}) \mapsto *$. 
Intuitively, while $\zeta^{\sharp}\colon C \to \nu T$ (see
\autoref{sec:prelim} for the notation) maps a coalgebra state of~$C$
to its tree of state transformers,
$\overline{\zeta}^{\sharp}(\cpair{s}{x})\in \nu B$ executes all these
state transformers without interruption, beginning at~$s$ and feeding
the output state of each previous step to the next step, and outputs
the intermediate states reached in each step.  Applying $\cofun$ to
the final coalgebra $(\nu T, \ter)$, we obtain a
$B$-coalgebra $(\stimes{\nu T}, \ol\ter)$, and currying the unique
coalgebra homomorphism $\ol\ter^\sharp\colon \stimes{\nu T} \to \nu B$
yields the desired quotient map%
\begin{equation}\label{eq:str}
  \trc=\curry(\overline{\ter}^{\sharp}) \c \nu T \epito (\nu B)^{S}.
\end{equation}

\begin{proposition}\label{prop:trc-surjective}
  The map $\trc$ is surjective.
\end{proposition}

\begin{defn}
The \emph{trace semantics} of a \isos{} specification $\L$ is given by
\[
  \sbrack{-}_{\L} = ( \mu\Sigma\xto{\brack{-}_\L} \nu T \xto{\trc}
  (\nu B)^S).
\]
Let $\simeq_\L$ denote the corresponding behavioural equivalence, that
is, $p\simeq_\L q$ iff $\sbrack{p}_\L = \sbrack{q}_\L$, for
$p,q\in \mu\Sigma$.  We drop
  subscripts if $\L$ is clear from the context.
\end{defn}


\begin{rem}
  Equivalently, $\sbrack{-}_\L$ is the curried form of the unique
  $B$-coalgebra homomorphism from $(\stimes{\mu\Sigma},\bar{\gamma})$
  to $\nu B$ (recall that $(\mu\Sigma,\gamma)$ is the operational
  model of~$\L$). Since
  \[
    \bar{\gamma}
    =
    \big(\stimes{\mu \Sigma} \xto{\gamma_0} \stimes{(\mu\Sigma+1)}
    \xto{\sbrks{\strength}} \stimes{(\stimes{\mu\Sigma}+1)} =
    B(\stimes{\mu\Sigma})\big)
  \]
  by definition of $\bar{\gamma}$ and \autoref{P:curry-h}, we see that
  for every $p\in \mu\Sigma$ and $s\in S$, the possibly infinite
  stream $\sbrack{p}_\L(s)=s_1s_2s_3\cdots$ is the sequence of states
  computed by the program $p$ on input state $s$,
  cf. \autoref{def:transition}:
  \[ s,p\to s_1,p_1\to s_2,p_2\to s_3,p_3 \to \cdots.\] Hence trace
  equivalence $p \simeq q$ holds iff for each input state $s$,
  programs $p$ and $q$ produce the same sequence of states.
\end{rem}
The following example demonstrates that trace semantics is generally
not compositional:

\begin{example}
  \label{eq:blatant}
  We extend \while{} by adding a unary operator $\lfloor \cdot \rfloor$ with
  \begin{equation*}
    \inference{\goes{s,p}{s\pr,p\pr}}{\goes{s,\lfloor p \rfloor}{\emptyset,\lfloor p\pr
        \rfloor}}
    \qquad
    \inference{\rets{s,p}{s\pr}}{\rets{s,\lfloor p \rfloor}{s\pr}}
  \end{equation*}
  where $\emptyset$ denotes the store with all variables set to
  $0$. For
  $t_{1} = \big(\mathtt{x} \mathbin{\coloneqq} 1\texttt{;}\ \mathtt{x}
  \mathbin{\coloneqq} \mathtt{x} + 1\big)$ and
  $t_{2} = \big(\mathtt{x} \mathbin{\coloneqq} 1\texttt{;}\ \mathtt{x}
  \mathbin{\coloneqq} \mathtt{x} * 2\big)$, we have that
  $t_{1} \simeq t_{2}$ but
  $\lfloor t_{1} \rfloor \not\simeq \lfloor t_{2} \rfloor$ (since
  in~$\lfloor t_{1} \rfloor$ and $\lfloor t_{2} \rfloor$, the store is
  erased after the first assignment).
\end{example}


\subparagraph{Termination Semantics}
\label{sec:quot2}

As the coarsest of our semantic domains, we shall use the set
\mbox{$(S \coproduct \{\bot\})^{S} \cong (S \coproduct \term)^{S}$}
of state transformers on~$S$ with possible non-termination featuring
pervasively in the denotational semantics of imperative programming
(e.g.~\cite{DBLP:journals/ita/Rutten99, Pitts:1999:ORF:309656.309671,
  DBLP:conf/ac/Pitts00}). In comparison to $(\nu B)^S$, this domain
abstracts from the intermediate steps of the computation.
The essence of this abstraction is captured by the map
\[\oname{fn} \c \nu B \to S \coproduct \term\qquad\text{defined by}\qquad
  \oname{fn} (x) =
  \begin{cases}
    s & \text{if $x$ is finite, with last state $s$},\\
    \bot & \text{otherwise}.
  \end{cases}
\]
\takeout{ 
Note that $\oname{fn}$ can be realized in the following pullback square:
\[
  \begin{tikzcd}
    \mu B
    \pullbackangle{-45}
    \arrow[r, "\oname{final}"]
    \arrow[d, hook, "\oname{inc}"']
    & S \arrow[d,  "\oname{inl}"] \\
    \nu B \arrow[r, "\oname{fn}"] & S \coproduct \term
  \end{tikzcd}
\]
where the initial algebra $\mu B$ is carried by the set of all
non-empty terminating sequences in $S$, $\oname{inc}$ is the inclusion
map, $\oname{inl}$ is the coproduct injection and
$\oname{final} \c \mu B \to S$ maps a terminating stream to its
final state.
\hunote{I don't quite see the relevance of this
  observation. SM: I agree and therefore I took this out.}
}

\begin{defn}
  The \emph{termination semantics} of a \isos{} specification $\L$ is
  given by
  \[
    \ssbrack{-}_{\L} = ( \mu\Sigma\xto{\sbrack{-}_\L} (\nu B)^S
    \xto{\fn^S} (S+1)^S).
  \]
  Let $\approx_\L$ denote the corresponding behavioural equivalence,
  that is, $p\approx_\L q$ iff $\ssbrack{p}_\L = \ssbrack{q}_\L$ for
  $p,q\in \mu\Sigma$.  We drop
  subscripts if $\L$ is clear from the context.
\end{defn}
Thus $p \approx q$ iff for each initial state $s$, if $p$ eventually
terminates with final state $s\pr$ then~$q$ eventually terminates with
final state $s\pr$ and vice-versa. Termination semantics is generally not
compositional: the programs $t_1$ and $t_2$ of
\autoref{eq:blatant} satisfy
$t_{1} \approx t_{2}$ but $\lfloor t_{1} \rfloor \not\approx \lfloor t_{2} \rfloor$.

The maps introduced in this section are summarized in the following commutative diagram:
\begin{equation}\label{eq:semantic-maps}
  \begin{tikzcd}
    &
    \mu\Sigma
    \ar[bend right]{dl}[swap,near end,inner sep=1]{\brack{-}_\L}
    \ar{d}{\sbrack{-}_\L}
    \ar[bend left,near end, inner sep=0]{dr}{\ssbrack{-}_\L}
    \\
    \nu T \ar[two heads]{r}{\trc}
    &
    (\nu B)^{S} \ar[two heads]{r}{\fn^{S}}
    &
    (S+1)^{S}
  \end{tikzcd}
\end{equation}

\subsection{Compositionality is Undecidable}
\label{sec:undec}

We have seen that in contrast to resumption semantics, both trace and
termination semantics generally fail to be compositional. As it turns
out, reasoning about compositionality in these two cases is a very
complex, viz.~undecidable, task.

To make the ensuing decision problems precise, we fix suitable
encodings of states and terms as finite strings and regard a \isos{} specification $\L$ as a total
function that assigns to a given operation symbol, input state and
list of premisses the target of the conclusion and output state of the
respective rule. From a computational point of view, a minimum
requirement on every reasonable specification $\L$ is that it admits
some finite representation. Hence, for
simplicity, we assume in the following theorem that specifications are
primitive recursive functions. For instance, this is clearly the case
for the \while{} language.
\begin{theorem}\label{thm:compositionality-undecidable}
  It is undecidable whether the trace semantics (or termination
  semantics, respectively) induced by a primitive recursive \isos{}
  specification is compositional.
\end{theorem}

\begin{proof}[Proof sketch]
The halting problem reduces to the compositionality problem. The idea is to take programs akin to $t_1$ and $t_2$ in \autoref{eq:blatant} and precompose them with the simulation of a given Turing machine. This can be specified in \isos{}. The failure of compositionality described in \autoref{eq:blatant} then occurs if, and only if, the simulated machine halts.
\end{proof}
\noindent In view
of the fact that there is no sound and complete decision procedure for
compositionality w.r.t.~$\simeq$ and $\approx$, we instead move on to
identify easily checked \emph{sound} syntactic criteria that, although
necessarily incomplete, are sufficiently broad.

\section{Cooling the \Isos{} Format}\label{sec:cool}

We now introduce two sets of restrictions on the \isos{} rule format,
called \emph{streamlined \isos{}} and \emph{cool \isos{}}, that
guarantee trace and termination semantics, respectively, to be
compositional. Our approach is inspired by the work of
Bloom~\cite{DBLP:journals/tcs/Bloom95} and van
Glabbeek~\cite{DBLP:journals/tcs/Glabbeek11} on the \emph{cool}
congruence formats for weak bisimilarity for GSOS specifications. The
following definition will help describe the restricted formats. We
make pervasive use of the abbreviations from \autoref{rem:omit}, and
we will additionally employ $\goes{s,p}{s',*}$ as an alternative
notation for a terminating literal $\rets{s,p}{s'}$.
\begin{defn}
  Let $\mathcal{L}$ be a \isos{} specification.
  \begin{enumerate}
  \item An $n$-ary operator $\f$ is
    \emph{passive} if all rules for $\f$ are of the form
    \[
      \inference{}{\goes{s,\f(x_{1},\dots,x_{n})}{s', t}} \qquad \text{where $t \in \Sigmas(\{x_{1},\dots,x_{n}\})$ or $t=\ast$.}
    \]
    In other words, the one-step behaviour of $\f(x_1,\ldots, x_n)$
    does not depend on the one-step behaviour of any of its subterms.
    In particular, every constant 
    is passive. An \emph{active} operator is one which is not passive.

  \item A progressing rule for an $n$-ary operator $\f$ is
    \emph{receiving at position $j\in \{1,\ldots, n\}$} if its~$j$-th
    premiss $\goes{s,x_{j}}{s',y_j}$ is progressing and the variable~$y_j$
    appears in the target of the conclusion. We say that the rule is
    \emph{receiving} if it is receiving at some  position $j$.
  \end{enumerate}
\end{defn}

\subsection{Streamlined \Isos{}}
\label{sec:streamlinedimpsos}
\noindent As indicated above, the streamlined \Isos{} format,
introduced next, will guarantee compositionality of trace
semantics.
\begin{defn}
  \label{def:sl}

  A \isos{} specification is \emph{streamlined} if for every active
  operator $\f$ of arity $n$ there exists $j\in \{1,\ldots,n\}$ (the
  \emph{receiving position} of $\f$) such that the following holds:
  \begin{enumerate}
  \item All receiving rules for $\f$ are of the form
    \[
      \inference{\goes{s,x_{j}}{s',y_j}}{\goes{s,\f(x_{1},\dots,x_{n})}{s',t}}\qquad \text{where $t=\f(x_{1},\dots,x_{n})[y_j/x_{j}]$ or $t=y_j$;}
    \]
    here, $[u/x]$ denotes substitution of the variable $x$ by the term $u$.
  \item All non-receiving rules for $\f$ are of the form
      \[
        \inference{l_{1} && l_{2} && \cdots && l_{n}}{\goes{s,\f(x_{1},\dots,x_{n})}{s',t}}
      \qquad\text{where $t \in \Sigmas(\{x_{1},\dots,x_{n}\} \smin \{x_{j}\})$ or
      $t=\ast$}.\]

\end{enumerate}
\end{defn}

\noindent Note that in a \isos{} specification, receiving rules for an
active operator~$\f$ are receiving \emph{only} in the receiving
position of~$\f$. What~\autoref{def:sl} boils down to is that an
active operator can only progress its subterm at the receiving
position~$j$, leaving everything else unchanged and making sure that
the output state in the $j$-th premiss is correctly propagated, and
discards the $j$-th subterm once it terminates.

\begin{example}\label{ex:while-streamlined}
  The \while{} language (cf.~\autoref{fig:while1}) is streamlined. The
  only active operator is sequential composition $p\texttt{;}~q$. Its
  progressing rules are receiving in the left position, and upon
  termination the left subterm is discarded.
\end{example}
Further examples are discussed after \autoref{cor:str-cong}.


\begin{theorem}
  \label{th:sl-congruence}
  Trace semantics is compositional for streamlined \isos{} specifications.
\end{theorem}

\begin{proof}[Proof sketch]
  For $p,q\in \mu\Sigma$ and $k\in \Nat$ we put $p\simeq_k q$ if the
  programs $p$ and $q$ are $k$-step trace equivalent, that is, for
  every $s\in S$ the streams $\sbrack{p}(s)$ and $\sbrack{q}(s)$ have
  the same prefix of length at most $k$. By induction on $k$ one
  proves~$\simeq_k$ to be a congruence, using a judicious
  strengthening of the inductive claim for receiving positions of
  active operators. This implies that $\simeq$ is a congruence, whence
  trace semantics is compositional.
\end{proof}
From \autoref{th:sl-congruence} we can deduce a slightly stronger
statement. In what follows, the \emph{kernel} of a map
$e\colon X\to Y$ is the equivalence relation on $X$ relating $x,x'$ iff $e(x)=e(x')$.%
\begin{corollary}\label{cor:str-cong}
  For every streamlined \isos{} specification, the kernel of the map
  $\trc\colon \nu T \epito (\nu B)^S$ is a congruence w.r.t.~the canonical $\Sigma$-algebra structure on $\nu T$ as per \eqref{eq:alpha}.
\end{corollary}
We next look at examples of streamlined specifications%
\lsnote{It would be good to have a realistic positive example}
but also at a few pathological cases where compositionality breaks.

\takeout{ 
\begin{example}
  Consider the \isos{} specification~$\L$ extending \while{}
  with a
  binary operator~$\Box$ given by
  \[
    \inference{\goes{s,p}{s\pr,p\pr} &&
      P(s,s\pr)}{\goes{s,p\mathbin{\square}q}{s\pr,p\pr\mathbin{\square}q}} \qquad
    \inference{\goes{s,p}{s\pr,p\pr} && \neg P(s,s\pr)}{\rets{s,p\mathbin{\square}q}{s_0}} \qquad
    \inference{\rets{s,p}{s\pr}}{\rets{s,p\mathbin{\square}q}{s\pr}}
  \]
  for some fixed state~$s_0\in S$. Then~$\L$ is streamlined: The
  operator~$\Box$ is active, with receiving position~$1$, and all
  rules adhere to the relevant restrictions. 
\end{example}}
\begin{example}
  \label{ex:good-sl}
  Streamlined specifications allow for complex control flow over programs,
  including \emph{signal} or \emph{interrupt handling}. For instance, we can
extend \while{} by a distinguished variable
  $\texttt{i}$ serving as an interrupt flag and modify the rules of sequential composition to
  \[
    \begin{array}{l@{\qquad}l@{\qquad}}
      \inference{\rets{s, p}{s\pr}}{\goes{s, (p\texttt{;}\ q)}{s\pr , q}}
      &
      \inference{\goes{s, p}{s\pr, p\pr}}{\goes{s, (p\texttt{;}\ q)}{s\pr ,
        (p\pr \texttt{;}\ q)}}{[\mathtt{i}]_{s} = 0} \\\\
      \inference{\goes{s, p}{s\pr, p\pr}}{\goes{s, (p\texttt{;}\ q)}{s\pr ,
      q}}{[\mathtt{i}]_{s} \neq 0 \wedge P(s\pr)}
      &
        \inference{\goes{s, p}{s\pr, p\pr}}{\goes{s, (p\texttt{;}\ q)}{s\pr ,
        (p\pr \texttt{;}\ q)}}{[\mathtt{i}]_{s} \neq 0 \wedge \neg P(s\pr)}
    \end{array}
  \]
  where $P \subseteq S$. If flag $\texttt{i}$ is enabled and predicate $P$
  is true for the output $s\pr$ of $p$, then $p$ is terminated prematurely. This type of rules can
  also be used to implement \emph{listeners} or \emph{observers} in high-level
  programming languages~\cite{JeffreyR05}.
\end{example}
\begin{example}
  \begin{enumerate}
  \item\label{ex:bad1} Recall the operator~$\lfloor \cdot \rfloor$
    from \autoref{eq:blatant}, which breaks compositionality for trace
    semantics. The operator is active, and its progressing rule is
    receiving but does not propagate the output state of its premiss,
    so the \isos specification of \while{}
    with~$\lfloor \cdot \rfloor$ fails to be streamlined (as it must,
    by \autoref{th:sl-congruence}).

  \item\label{item:interleave} Consider the extension of \while{} with a binary left-first
    interleaving operator $\triangleleft$ specified by the
    rules
    \[
      \inference{\goes{s,p}{s\pr,p\pr}}{\goes{s,p \triangleleft q}{s\pr,q
          \triangleleft p\pr}}
      \qquad
      \inference{\rets{s,p}{s\pr}}{\goes{s,p \triangleleft q}{s\pr,q}}
    \]
    Again, $\simeq$ is not a congruence: For
    $t_{1} = (\mathtt{x} \mathbin{\coloneqq} 2 \texttt{;}\ \mathtt{x}
    \mathbin{\coloneqq} \mathtt{x} + 2)$ and
    $t_{2} = (\mathtt{x} \mathbin{\coloneqq} 2 \texttt{;}\ \mathtt{x}
    \mathbin{\coloneqq} \mathtt{x} * 2)$, we have $t_{1} \simeq t_{2}$
    but
    $t_{1} \triangleleft (\mathtt{x} \mathbin{\coloneqq} 0) \not
    \simeq t_{2} \triangleleft (\mathtt{x} \mathbin{\coloneqq}
    0)$. Indeed the left of the above rules is receiving but the
    target of its conclusion does not have one of the allowed forms.

  \item Extend \while{} with a step-by-step branching operator
    $\mathbin{\triangledown}$ specified by
    \[
      \inference{\goes{s,p}{s_{1},p\pr} && \goes{s,q}{s_{2},q\pr}}{\goes{s,p \mathbin{\triangledown} q}{s_{1},p\pr
          \mathbin{\triangledown} q}}{P(s)}
      \qquad
      \inference{\goes{s,p}{s_{1},p\pr} && \goes{s,q}{s_{2},q\pr}}{\goes{s,p
          \mathbin{\triangledown} q}{s_{2},p \mathbin{\triangledown} q\pr}}{\neg P(s)}
    \]
    and termination in all other cases. If the predicate
    $P \subseteq S$ is, for example, $\mathtt{x} = 0$, then the
    same~$t_1,t_2$ as in \autoref{item:interleave} witness that
    $\simeq$ is not a congruence: We have $t_{1} \simeq t_{2}$ but
    $t_{1} \mathbin{\triangledown} (\mathtt{x} \mathbin{\coloneqq} 0)
    \not \simeq t_{2} \mathbin{\triangledown} (\mathtt{x}
    \mathbin{\coloneqq} 0)$. In this case, the condition that is
    violated is the requirement that all rules for~$\triangledown$ must be
    receiving in the same position.

  \item Consider the operator
    $\lceil \cdot \rceil$ specified by
    \[
      \inference{\goes{s,p}{s\pr,p\pr}}{\goes{s,\lceil p \rceil}{s\pr,\lceil p\pr
          \rceil}}
      \qquad
      \inference{\rets{s,p}{s\pr}}{\goes{s,\lceil p \rceil}{s\pr,p}}
    \]
    Again, $t_{1},t_{2}$ as in \autoref{item:interleave} witness
    failure of congruence: $t_{1} \simeq t_{2}$ but
    $\lceil t_{1} \rceil \not\simeq \lceil t_{2} \rceil$. Indeed, the
    second rule violates~\autoref{def:sl} as~$p$ terminates but is not
    discarded.
  \end{enumerate}
\end{example}

\subsection{Cool \isos{}}
\label{sec:faimpsos}
We now further restrict the streamlined format as follows:
\begin{defn}
  \label{def:cool}
  A \isos{} specification is \emph{cool} if for every active
  operator $\f$ there exists $j\in \{1,\ldots, n\}$ (again called the \emph{receiving position of $\f$}) such that the following holds:
  \begin{enumerate}
    \item\label{item:patience} All rules for $\f$ whose $j$-th premiss is progressing are of the form
      \[
        \inference{\goes{s,x_{j}}{s',y_j}}{\goes{s,\f(x_{1},\dots,x_{n})}{s',\f(x_{1},\dots,x_{n})[y_j/x_{j}]}}
      \]
    \item All rules for $\f$ whose $j$-th premiss is terminating are of the form
      \[
        \inference{\rets{s,x_{j}}{s'}}{\goes{s,\f(x_{1},\dots,x_{n})}{s'',t}}\qquad\text{where $t\in \Sigmas(\{x_{1},\dots,x_{n}\} \smin \{x_{j}\})$ or $t=\ast$},
      \]
      and moreover $s''$ and $t$ depend only on $s'$ but not on $s$.

    \end{enumerate}
    A \isos specification is \emph{uncool} if it is not cool.
  \end{defn}
The cool format asserts that an active operator $\f$ runs its $j$-th
subterm until termination and then discards it, proceeding to a state
derivable from the terminating state of the subterm. In GSOS, rules of
type~\ref{item:patience} (without states) are known as
\emph{patience rules}~\cite{DBLP:journals/tcs/Glabbeek11}.

\begin{example}
  The rules of the \while{} language, which we have already observed
  to be streamlined (\autoref{ex:while-streamlined}), are also cool.
%
\end{example}

\noindent Cool \isos{} specifications are streamlined, and all of the
negative examples from \autoref{sec:streamlinedimpsos} apply here as
well. Here is an example that separates the two concepts:

\begin{example}\label{ex:bad5}
  The sequential composition semantics with interrupts from \autoref{ex:good-sl} is
  uncool, as the third rule has a progressing premiss but is not of
  the form in \autoref{def:cool}.\ref{item:patience}. 
  Indeed, $\approx$ is not a congruence:
  For the predicate $\mathtt{x} = 42$ and the programs
  $t_{1} = (\mathtt{x} \mathbin{\coloneqq} 42\texttt{;}\ \mathtt{x} \mathbin{\coloneqq}
  2)$ and $t_{2} = (\mathtt{x} \mathbin{\coloneqq} 2)$, we have
  $t_{1} \approx t_{2}$ but
  $t_{1}\texttt{;}\ \mathtt{skip} \not\approx
  t_{2}\texttt{;}\ \mathtt{skip}$.
\end{example}
\noindent As indicated above, coolness guarantees congruence for
termination semantics:

%





\begin{theorem}
  \label{th:fa-congruence}
 Termination semantics is compositional for cool \isos{} specifications.
\end{theorem}

\begin{proof}[Proof sketch]
Suppose that $\f\in \Sigma$ is an $n$-ary operator and $p_m,q_m\in \mu \Sigma$ are programs with $p_m\approx q_m$ for $m=1,\ldots,n$. By symmetry, it suffices to show the following for all $s,\ol{s}\in S$:
\[
    \text{If $s,\f(p_1,\ldots,p_m)$ terminates in state
      $\ol{s}$, then $s,\f(q_1,\ldots,q_m)$ terminates in state
      $\ol{s}$.}
\]
The proof proceeds by an outer induction on the number of steps until
termination of $s,\f(p_1,\ldots,p_m)$ and an inner induction on the
structure of the programs.
\end{proof}

\noindent By \autoref{cor:str-cong} we know that for every cool (whence
streamlined) specification the kernel of
$\trc\colon \nu T\epito (\nu B)^S$ forms a congruence. Since $\trc$ is
surjective, this means precisely that there is a (unique) $\Sigma$-algebra
structure on $(\nu B)^S$ for which $\trc$ is a $\Sigma$-algebra
homomorphism.
\begin{corollary}\label{cor:fns-cong}
  For every cool \isos{} specification, the kernel of
  $\fn^S\colon (\nu B)^S\epito (S+1)^S$ is a congruence w.r.t.~the
  induced $\Sigma$-algebra structure on $(\nu B)^S$.
\end{corollary}
%
%
%
%

\section{Conclusions and Future Work}

We have introduced the \emph{\isos{}} rule format for the operational
semantics of stateful languages, and equipped it with three semantics:
resumption semantics, trace semantics, and termination semantics, in
decreasing order of granularity. Our main interest has been in
compositionality of these semantics. While resumption semantics is
always compositional, it is in general undecidable whether the coarser
semantics are compositional. However, compositionality is ensured by
restricting to \emph{streamlined} \isos{} specifications for trace
semantics, and to \emph{cool} \isos{} specifications for termination
semantics. The compositionality result for the cool format improves on
previous results for the similar \emph{evaluation-in-context}
formats~\cite{DBLP:conf/fossacs/Abou-SalehP13} by abstracting from
steps until termination. The streamlined format is more permissive, as
we illustrate on a signal handling construct.

Our results currently work with deterministic state transformers,
captured by the functor $TX = (BX)^S$ where $BX = \stimes{(X+1)}$. We
believe that our results generalize to functors~$B$
equipped with a natural transformation $c_X\colon BX \to S$. As a first step, this generalization requires an abstract characterization of our streamlined and cool rule formats in terms of their corresponding natural transformations, along with categorical proofs of the respective congruence theorems. We leave this as an important point for future work.

A further direction of possible generalization is to cover
\emph{effects}, such as non-determinism, in a similar style as in work
on evaluation-in-context~\cite{DBLP:conf/fossacs/Abou-SalehP13}. Our
work embeds the standard semantics of sequential imperative
programming (in particular termination semantics) into the paradigm of
operational semantics via distributive laws, and we expect to relate
our results to work on morphisms of distributive
laws~\cite{DBLP:journals/entcs/Watanabe02,DBLP:conf/calco/KlinN15},
which, for instance, have recently been shown to have applications to
secure compilation~\cite{DBLP:conf/cmcs/0001NDP20}. Extending the
overall paradigm to support higher-order languages is a well-known
and, so far, elusive problem. Like in the current work, tackling this
problem may require a slight deviation from the standard form of GSOS
laws. It is worth noting that rule formats for higher-order languages
have been proposed in the past by
Howe~\cite{DBLP:journals/iandc/Howe96},
Bernstein~\cite{DBLP:conf/lics/Bernstein98} and more recently
Hirschowitz and Lafont~\cite{DBLP:journals/corr/abs-2103-16833}.

Our treatment of resumption and trace semantics and their relationship
is generic, and presumably can be transferred to other settings, in
particular to constructive and type-theoretic frameworks. Indeed we
expect that it can be implemented relatively directly in
foundational proof assistants such as Agda, without additional
postulates (such as the axiom of choice or the law of excluded
middle). In contrast, the domain $(S+1)^S$ of termination semantics is
inherently classical, as it postulates that every computation will
either terminate or diverge. This can be remedied by replacing the
maybe-monad $(-)+1$ with a suitable \emph{partiality
  monad}~\cite{AltenkirchDanielssonEtAl17,ChapmanUustaluEtAl19}. We
will explore to what extent our results regarding
termination semantics can be rebased on this more general perspective.







\bibliography{mainBiblio}

\clearpage
\appendix
\section{Appendix: Omitted Proofs}

\subsection*{Proof of~\autoref{thm:specifications-vs-transformations}}

Let $B=\stimes{(\Id+1)}$. The correspondence between~(2) and~(3) is proven by the following bijections:
\begin{align*}
  & \NT(\stimes{\Sigma(\Id\times B)},B\Sigmas)\\
  & \iso \;\NT\Bigl(\sum\nolimits_{\f\in\Sigma} \stimes{(\Id\times B)^{\ar(\f)}},B\Sigmas\Bigr)\\
  &\iso\;\prod\nolimits_{\f\in\Sigma}\NT(\stimes{(\Id\times B)^{\ar(\f)}},B\Sigmas)\\
  &\iso\;\prod\nolimits_{\f\in\Sigma}\NT(\stimes{(\Id\times \stimes{(\Id+\term)})^{\ar(\f)}},B\Sigmas)\\
  &\iso\;\prod\nolimits_{\f\in\Sigma}\NT(\stimes{S^{\ar(\f)}\times \Id^{\ar(\f)}\times  (\Id+\term)^{\ar(\f)}},B\Sigmas)\\
  &\iso\;\prod\nolimits_{\f\in\Sigma}\NT(\Id^{\ar(\f)}\times (\Id+\term)^{\ar(\f)},(B\Sigmas)^{ \stimes{S^{\ar(\f)}}})\\
  &\iso\;\prod\nolimits_{\f\in\Sigma}\NT\Bigl(\Id^{\ar(\f)}\times \sum\nolimits_{W\subseteq\ar(\f)} \Id^{W},(B\Sigmas)^{ \stimes{S^{\ar(\f)}}}\Bigr)\\
  &\iso\;\prod\nolimits_{\f\in\Sigma}\prod\nolimits_{W\subseteq\ar(\f)}\NT(\Id^{\ar(\f)}\times \Id^{W},(B\Sigmas)^{ \stimes{S^{\ar(\f)}}})\\
  &\iso\;\prod\nolimits_{\f\in\Sigma}\prod\nolimits_{W\subseteq\ar(\f)}\NT(\Id^{\ar(\f)+W},(B\Sigmas)^{\stimes{S^{\ar(\f)}}})\\*
  &\iso\;\prod\nolimits_{\f\in\Sigma}\prod\nolimits_{W\subseteq\ar(\f)} (B\Sigmas(\ar(\f)+W))^{\stimes{S^{\ar(\f)}}}
\end{align*}
The last step uses the Yoneda lemma, and all other steps follow from the
definition of (co-)products and the fact that products distribute over
coproducts in $\Set$.

To prove the correspondence between~(1) and~(2) we first observe that
every \isos{} specification $\L$ induces a family
\[
  \big(r_{\f,W}\c \stimes{S^{\ar(\f)}}\to \stimes{\Sigmas(\ar(\f)+W)}+
  S\big)_{\f\in\Sigma,W\subseteq \ar(\f)}
\]
as follows: for every $n$-ary $\f\in \Sigma$,
$s,s_1',\ldots,s_n'\in S$ and $W\seq \{1,\ldots,n\}$, if $\L$ contains
the rule
\[
  \inference{(s,x_j\to s_j',y_j)_{j\in W}\qquad(\rets{x_j}{s_j'})_{j\in\{1,\ldots,n\}\smin W}}{s,\f(x_1,\ldots,x_{n})\to s',t}, 
\]
then $r_{\f,W}(\cpair{s}{s_1',\ldots,s_n'})=(s',t')$ where
$t'\in \Sigmas(\ar(\f)+W)$ is obtained from $t$ by substituting the
variable $x_j$ by $j$ (in the left coproduct component of $\ar(\f)+W$) for each $j\in \ar(\f)$, and the variable $y_j$ by $j$ (in the right coproduct component of $\ar(\f)+W$) for each $j\in
W$. If $\L$ contains the rule
\[
  \inference{(s,x_j\to s_j',y_j)_{j\in W}\qquad(\rets{x_j}{s_j'})_{j\in\{1,\ldots,n\}\smin W}}{\rets{s,\f(x_1,\ldots,x_{n})}{s'}}. 
\]
then $r_{\f,W}(\cpair{s}{s_1',\ldots,s_n'})=s'$. This translation
clearly defines a bijective correspondence between \isos{}
specifications $\L$ and families $(r_{\f,W})_{\f,W}$.

The construction $\L\mapsto \delta$ in
\autoref{rem:spec-vs-trafo}.\ref{rem:spec-to-trafo} is obtained by
first forming the family $(r_{\f,W})_{\f,W}$ corresponding to $\L$ and
then taking the \isos{} law $\delta$ corresponding to
$(r_{\f,W})_{\f,W}$ according to the above isomorphism.

Similarly, the construction $\delta\mapsto \L$ in
\autoref{rem:spec-vs-trafo}.\ref{rem:trafo-to-spec} is obtained by
first forming the family $(r_{\f,W})_{\f,W}$ corresponding to $\delta$
according to the above isomorphism, and then taking the \isos{}
specification $\L$ corresponding to $(r_{\f,W})_{\f,W}$.

Consequently, the two constructions in \autoref{rem:spec-vs-trafo} are
mutually inverse.\qed

\subsection*{Proof of \autoref{P:curry-h}}

\begin{rem}
  The proof of \autoref{P:curry-h} makes use of an induction principle
  which combines two induction principles known in the literature:
  \emph{primitive recursion} (see
  e.g.~\cite[Prop.~2.4.7]{DBLP:books/cu/J2016}) and \emph{induction
    with parameters} (see
  e.g.~\cite[Exercise.~2.5.5]{DBLP:books/cu/J2016}. We recall
  primitive recursion in \autoref{R:princ-1} below and use it in
  \autoref{R:princ-2} to establish the combined principle we need.
  \begin{enumerate}
  \item\label{R:princ-1} Let $F$ be a functor with an initial algebra $\mu F$ on a category
    with finite products. Primitive recursion states that for every
    morphism $\alpha\colon F(\mu F \times A) \to A$
    there exists a unique morphism
    $h\colon \mu F \to A$ such that the following square commutes:
    \[
    \begin{tikzcd}
      F(\mu F)
      \ar{r}{\ini}
      \ar{d}[swap]{F\brks{\id, h}}
      &
      \mu F
      \ar{d}{h}
      \\
      F(\mu F \times A)
      \ar{r}{\alpha}
      &
      A
    \end{tikzcd}
  \]

\item\label{R:princ-2} Now assume in addition that $F$ is strong and that the base
  category is cartesian closed. Then for every object $Y$ and every
  morphism $\beta\colon \ytimes{F(\mu F \times A)} \to A$ there exists
  a unique morphism $h\colon \ytimes{\mu F} \to A$ such that the
  following diagram commutes:
    \begin{equation}\label{eq:strong-prim}
      \begin{tikzcd}
        \ytimes{F(\mu F)}
        \ar{r}{\tpair{\id}{\ini}}
        \ar{d}[swap]{\sbrks{\strength}}
        &
        \ytimes{\mu F}
        \ar{dd}{h}
        \\
        \ytimes{F(\ytimes{\mu F})}
        \ar{d}[swap]{\tpair{\id}{F\brks{\snd,h}}}
        \\
        \ytimes{F(\mu F \times A)}
        \ar{r}{\beta}
        &
        A
      \end{tikzcd}
    \end{equation}
    Indeed, given $\beta$ one obtains a morphism $\alpha\colon F(\mu F
    \times A^Y) \to A^Y$ by currying the morphism
    \begin{equation}
      \begin{aligned}
        \ytimes{F(\mu F \product A^{Y})}
        &
        \xra{\hspace*{3.7mm}\sbrks{\strength}\hspace*{3.7mm}}
        \ytimes{F(\ytimes{\mu F \product A^{Y}})} \cong \ytimes{F({\mu F\product Y \product A^{Y}})}
        \\
        &
        \xra{\tpair{\id}{F(\id \product \ev )}}
        \ytimes{F (\mu F \product A)}
        \xra{\beta}
        A.
      \end{aligned}
    \end{equation}
    By primitive recursion there exists a unique morphism
    $h\colon \ytimes{\mu F} \to A$ such that the square below
    commutes:
    \begin{equation}
      \label{eq:str1}
      \begin{tikzcd}[column sep = 30]
        F (\mu F)
        \ar{r}{\ini}
        \ar{d}[swap]{F\langle \id, \curry(h) \rangle}
        &
        \mu F
        \ar{d}{\curry(h)}
        \\
        F(\mu F \product A^{Y})
        \arrow{r}{\alpha}
        &
        A^{Y}
      \end{tikzcd}
    \end{equation}
    Now consider the diagram below:
    \begin{equation}
      \label{eq:str2}
      \begin{tikzcd}[column sep = 60]
        \ytimes{F (\mu F)}
        \ar{rr}{\tpair{\id}{\ini}}
        \ar{rd}[inner sep = 1]{\tpair{\id}{F\langle \id,\curry(h) \rangle}}
        \ar{d}[swap]{\sbrks{\strength}}
        &
        &
        \ytimes{\mu F}
        \ar{ddd}{h}
        \\
        \ytimes{F(\ytimes{\mu F})}
        \ar{rd}[swap,near end]{\tpair{\id}{F(\brks{\id,\curry(h)})}}
        \ar{dd}[swap]{\tpair{\id}{F\brks{\snd,h}}}
        &
        \ytimes{F(\mu F \product A^{Y})}
        \ar{d}{\sbrks{\strength}}
        \\
        &
        \ytimes{F(\ytimes{\mu F \product A^{Y}})}
        \ar{ld}{\tpair{\id}{F(\id \product \ev )}}
        \\
        \ytimes{F (\mu F \product A)}
        \arrow{rr}{\beta}
        &&
        A
      \end{tikzcd}
    \end{equation}
    Its outside is~\eqref{eq:strong-prim}, the right-hand part is
    precisely the uncurrying of~\eqref{eq:str1}, and the remaining two
    inner parts clearly commute. Hence, the outside commutes iff so
    does the right-hand part which happens iff~\eqref{eq:str1}
    commutes. This proves the desired result.
  \end{enumerate}
\end{rem}
\begin{proof}[Proof of \autoref{P:curry-h}]
  In \autoref{def:transition}, the map
  $\gamma_0\colon \stimes{\mu\Sigma} \to \stimes{(\mu\Sigma +1)}$ is defined
  such that the following diagram commutes:
  \[
    \begin{tikzcd}[column sep = 45]
      \stimes{\Sigma (\mu \Sigma)}
      \ar{rr}{\tpair{\id}{\ini}}
      \ar{d}[swap]{\sbrks{\strength}}
      & &
      \stimes{\mu \Sigma}
      \ar{dd}{\gamma_0}
      \\
      \stimes{\Sigma(\stimes{\mu \Sigma})}
      \ar{d}[swap]{\tpair{\id}{\Sigma\brks{\snd,\gamma_0}}}
      \\
      \stimes{\Sigma (\mu \Sigma \product \stimes{(\mu\Sigma \coproduct \term)})}
      \ar{r}{\delta_{\mu \Sigma}}
      & \stimes{(\Sigmas(\mu\Sigma) \coproduct \term)}
      \ar{r}{\tpair{\id}{(\hat\ini \coproduct \id)}}
      & \stimes{(\mu\Sigma \coproduct \term)}
    \end{tikzcd}
  \]
This is an instance of~\eqref{eq:strong-prim} with $h=\gamma_0$ and $\beta = (\tpair{\id}{(\hat\ini+\id)}) \comp \delta_{\mu\Sigma}$.
  The map $\alpha$ in~\eqref{eq:str1} is the
  $\mu\Sigma$-component of the natural transformation
  in~\eqref{eq:embed} composed with $\tpair{\id}{(\hat\ini + \id)}$. Thus currying this yields
  $T\hat\ini \cdot \hat\delta_{\mu\Sigma}$
  (cf.~\eqref{eq:gamma}). This implies the desired result.
\end{proof}

\subsection*{Proof of \autoref{prop:trc-surjective}}
  \lsnote{This part I'm a bit worried about. When one understands
    what $\nu T$ and $\nu B$ are, then the fact is just blatantly
    obvious, and in no way justifies such a complicated proof. (Here
    is a short one: Given an element $f$ of $\nu B^S$, build a
    tree~$t$ in $\nu T$ that has, for each~$s\in S$, a tree $t_s$ as
    the $s$-th child of the root, where $t_s$ is a tree that has
    uniform depth being the length of the sequence $f(s)$, and labels
    all edges at level~$i<\omega$ with the $i$-th entry of
    $f(s)$. Then $\trc(t)=f$.)
    SM: I do not immediately understand this argument (e.g.~uniform depth is
    not defined), but I find the categorical proof completely lucid. It
    also substantiates the claim in our conclusion that our results
    hold for more general functors $B$. So it's good to have the
    categorical proof. Reader for whom this is clear may simply skip
    reading the proof.} The object $(\nu B)^S$ can be equipped
  with the $T$-coalgebra structure
  \begin{align*}
    (\nu B)^S\xto{\beta^S} (B(\nu B))^S = T(\nu B) \xto{T\Delta} T((\nu B)^S),
  \end{align*}
  where $\beta$ is the structure of the final coalgebra $\nu B$ and
  $\Delta=\brks{\id}_{s\in S}$ denotes the diagonal. We claim that the
  unique $T$-coalgebra homomorphism
  \[
    m = (T\Delta\comp \beta^S)^{\sharp}\colon (\nu B)^S\to \nu T
  \]
  is a splitting of $\trc$ (that is, $\trc\comp m = \id$), which
  implies that $\trc$ is surjective. To prove the claim, we consider the
  following $B$-coalgebra structure on $\stimes{(\nu B)^S}$:
  \[
    \stimes{(\nu B)^S}
    \xto{\ev}
    \nu B
    \xto{\beta}
    B(\nu B)
    \xto{\sbrks{B\Delta}}
    \stimes{B ((\nu B)^S )}
    \xto{\strength}
    B(\stimes{(\nu B)^S}).
  \]
  We will show below that both
  $\ev\colon \stimes{(\nu B)^S} \to \nu B$ and
  $\bar{\tau}^{\sharp} \comp (\stimes{m})$ are
  homomorphisms from this coalgebra to $\nu B$. Then
  $\bar{\tau}^{\sharp} \comp (\stimes{m}) = \ev$ by finality, whence
  the desired result holds as follows:
  \[
    \trc\comp m = \curry(\bar{\tau}^{\sharp})\comp m = \curry(\ev) =
    \id.
  \]
  The commutative diagram below proves $\ev$ to be a coalgebra homomorphism:
  \[
    \begin{tikzcd}
      \stimes{(\nu B)^S}
      \ar{r}{\ev}
      \ar{d}[swap]{\ev}
      &
      \nu B
      \ar{r}{\beta}
      &
      B(\nu B)
      \ar{drr}{\id}
      \ar{r}[rotate=20,pos=1,inner sep = 7]{ \sbrks{B\Delta} }
      &
      \stimes{B( (\nu B)^S )}
      \ar{r}{\strength}
      &
      B ( \stimes{(\nu B)^S} )
      \ar{d}{B\ev}
      \\
      \nu B \ar{rrrr}{\beta}
      &&&&
      B(\nu B)
    \end{tikzcd}
  \]
  To see that the right-hand triangle commutes, let $(\cpair{s}{w})\in \stimes{(\nu B+1)} = B(\nu B)$. If
  $w\in \nu B$, then we have
  \[
    (\cpair{s}{w})
    \xmapsto{\sbrks{B\Delta}}
    (\cpair{s}{(\cpair{s}{c_w})})
    \xmapsto{\strength}
    (\cpair{s}{(\cpair{s}{c_w})})
    \xmapsto{B\ev}
    (\cpair{s}{w})
  \]
  where $c_w\in (\nu B)^S$ is the constant map with value $w$. If
  $w=*$, we have
  \[
    (\cpair{s}{w})=(\cpair{s}{*})
    \xmapsto{\sbrks{B\Delta}}
    (\cpair{s}{(\cpair{s}{*})})
    \xmapsto{\strength}
    (\cpair{s}{*})
    \xmapsto{B\ev}
    (\cpair{s}{*})=(\cpair{s}{w}).
  \]
  Similarly, the diagram below shows that
  $\bar{\tau}^{\sharp}\comp (\stimes{m})$ is a coalgebra homomorphism:
  \[
    \begin{tikzcd}[column sep = 15]
      \stimes{(\nu B)^S}
      \ar{r}{\ev}
      \ar{ddd}[swap]{\stimes{m}}
      \ar{dr}[inner sep=0]{\stimes{\beta^S}}
      &
      \nu B
      \ar{r}{\beta}
      &
      B(\nu B)
      \ar{r}{\sbrks{B\Delta}}
      \ar{rdd}[swap,inner sep=1]{\sbrks{B(m\cdot \Delta)}}
      &[6ex]
      \stimes{B((\nu B)^S)}
      \ar{dd}{\stimes{Bm}}
      \ar{r}{\strength}
      &
      B ( \stimes{(\nu B)^S} )
      \ar{ddd}{B(\stimes{m})}
      \\
      &
      \stimes{(B(\nu B))^S}
      \ar{ru}{\ev}
      \ar{d}[swap]{\stimes{T\Delta}}
      \ar[bend left=40,xshift=20]{dd}{\stimes{T(m\cdot\Delta)}}
      \\
      &
      \stimes{T((\nu B)^S)}
      \ar{d}[swap]{\stimes{Tm}}
      &&
      \stimes{B(\nu T)}
      \ar{rd}{\strength}
      \\
      \stimes{\nu T}
      \ar{r}{\stimes{\tau}}
      \ar{d}[swap]{\bar{\tau}^{\sharp}}
      \ar[yshift=-2pt,bend right = 6]{rrrr}{\ol{\tau}}
      &
      \stimes{T(\nu T)}
      \ar{rr}{\ev}
      &&
      B(\nu T)
      \ar{u}{\sbrks{\id} }
      \ar{r}{\sbrks{\strength}}
      &
      B(\stimes{\nu T})
      \ar{d}{B(\bar{\tau}^{\sharp})}
      \\[1ex]
      \nu B \ar{rrrr}{\beta}
      & & & &
      B(\nu B)
    \end{tikzcd}
  \]
  Indeed, we show that all its inner parts commutes, whence so does
  the outside as desired. The lowest part commutes since
  $\ol\tau^\sharp$ is a coalgebra homomorphism, and the small part
  above it is the definition of $\ol\tau$. The upper right-hand part
  commutes by the naturality of
  $\strength\colon \stimes{BX} \to B(\stimes{X})$, and the triangles
  below it and on its left obviously do. The upper left-hand triangle
  commutes by the naturality of $\ev$, and the left-hand part does
  since $m\colon (\nu B)^S \to \nu T$ is a $T$-coalgebra
  homomorphism. For the remaining middle part we consider the
  components of the product $\stimes{B(\nu T)}$ separately. The
  right-hand component commutes due to the naturality of
  $\ev\colon \stimes{TX} = \stimes{(BX)^S} \to BX$, and for the
  left-hand component we expand the definition of $T$ and obtain the
  following commutative diagram:
  \[
    \begin{tikzcd}
      \stimes{(B(\nu B))^S}
      \ar{rr}{\ev}
      \ar{dd}[swap]{\tpair{\id}{(B(m\cdot\Delta))^S}}
      \ar{rd}[inner sep = 1]{\stimes{(\fst)^S}}
      &&
      B(\nu B) = \mathrlap{\stimes{(\nu B + 1)}}
      \ar{d}{\fst}
      \\
      &
      \stimes{S^S}
      \ar{r}{\ev}
      &
      S
      \\
      \stimes{(B(\nu T))^S}
      \ar{rr}{\ev}
      \ar{ru}[swap,inner sep=0]{\stimes{(\fst)^S}}
      &&
      B(\nu B) = \mathrlap{\stimes{(\nu B + 1)}}
      \ar{u}[swap]{\fst}
    \end{tikzcd}
  \]
  The two right-hand parts obviously commute, and the left-hand
  triangle also clearly does: remove $\stimes{(\argument)^S}$ and use
  that $B(m \cdot \Delta) = \stimes{(m \cdot \Delta + 1)}$.
  This completes the proof.\qed

\subsection*{Proof of \autoref{thm:compositionality-undecidable}}
We split \autoref{thm:compositionality-undecidable} into two separate results:

\begin{theorem}\label{thm:compositionality-undecidable1}
  It is undecidable whether the trace semantics induced by a primitive
  recursive \isos{} specification is compositional.
\end{theorem}

\begin{proof}
  We reduce the undecidable problem whether a given Turing machine $M$
  halts on the empty tape to the compositionality problem. Let $Q$ be
  the set of states and $A\cup \{{\sharp}\}$ the tape alphabet with ${\sharp}$
  representing empty cells. Recall that a \emph{configuration} of $M$
  is a pair $C=(q,v\underline{a}w)$ where $q\in Q$ is the current
  state, $vaw$ is the tape content and the underscore indicates the
  head position. We write $C\vdash C'$ if $C'$ is the successor
  configuration of $C$ according to the transitions of $M$. The
  initial configuration is $C_0=(q_0,\underline{{\sharp}})$, and we assume
  w.l.o.g.~that there is a unique halting configuration
  $C_{\mathrm{halt}}$.

From $M$ we construct a \isos{} specification $\L$ as follows. The signature is $\Sigma=\{i,j,c_0,c,d_0,d,u\}$ with $u$ unary and all other symbols constants, and the states are
\[ S=\mathrm{Conf}\cup \{0,1,\mathsf{err}\} \cup \Nat\times \{c_0,c,d_0, d\} \]
where $\mathrm{Conf}$ is the set of configurations of $M$. The rules are listed below; all cases where no explicit rule is specified lead to termination in the error state $\mathsf{err}$.
\begin{itemize}
\item Rules for $i$ and $j$:
\[
\frac{}{s,i\to C_0,c_0} \qquad \frac{}{s,j\to C_0,d_0}
\]
for all $s\in S$.
\item Rules for $c_0$ and $d_0$:
\[ \frac{}{C,c_0\to C',c_0} \qquad \frac{}{C_{\mathrm{halt}},c_0\to 0,c} \qquad \frac{}{(n,c_0),c_0\downarrow (n,c_0)}  \]
\[ \frac{}{C,d_0\to C',d_0} \qquad \frac{}{C_{\mathrm{halt}},d_0\to 0,d} \qquad \frac{}{(n,d_0),d_0\downarrow (n,d_0)}    \]
for all $C,C'\in\mathrm{Conf}$ with $C\vdash C'$ and $n\in \Nat$.
\item Rules for $c$ and $d$:
\[\frac{}{0,c\to 0,c} \qquad \frac{}{1,c\downarrow 1}\qquad \frac{}{(n,c),c\downarrow (n,c)}  \]
\[ \frac{}{0,d\to 0,d} \qquad \frac{}{1,d\downarrow 0} \qquad \frac{}{(n,d),d\downarrow (n,d)} \]
for all $n\in \Nat$.
\item Rules for $u$:
\[ \frac{s,p\to s',p'}{s,u(p) \to 1,p'}\qquad \frac{C,p\to s',p'}{C,u(p)\to s',u(p')} \qquad \frac{}{(n,k),u(p)\to (n+1,k),p}
\]
for all $s\in \{0,1\}$, $s'\in S$, $C\in \mathrm{Conf}$ and $(n,k)\in \Nat\times \{c_0,c,d_0,d\}$.
\end{itemize}
Clearly $\L$ has a primitive recursive representation derivable from the description of the given machine $M$. For the induced trace equivalence $\simeq\,=\,\simeq_\L$ we now show that
\[
  \text{$M$ halts on the empty tape}
  \qquad\text{iff}\qquad
  \text{$\simeq$ is not a congruence}.
\]
From the undecidability of the halting problem on the empty tape it
then follows that the congruence property of $\simeq$,
i.e.~compositionality of $\L$ under trace semantics, is undecidable.

\medskip\noindent ($\Longrightarrow$) Suppose that $M$ halts on the empty
tape, and let $C_0\vdash C_1\vdash\ldots\vdash C_n=C_{\mathrm{halt}}$ be the finite run
of $M$.  Then $i\simeq j$ but $u(i)\not\simeq u(j)$. Indeed, running
$i$ and $j$ on any input $s$ yields
\[
  s,i\to C_0, c_0\to C_1,c_0\to \cdots\to
  C_n=C_{\mathrm{halt}},c_0\to 0,c\to 0,c \to \cdots
\]
and
\[
  s,j\to C_0,d_0\to C_1,d_0\to \cdots\to  C_n=C_{\mathrm{halt}},d_0\to
  0,d\to 0,d\to\cdots,
\]
and so $i\simeq j$. On the other hand, running $u(i)$ and $u(j)$ on input $C_0$ yields
\[
  C_0,u(i)\to C_0, u(c_0)\to  C_1, u(c_0)\to \cdots\to
  C_n=C_{\mathrm{halt}}, u(c_0)\to  0,u(c)\to 1,c \downarrow 1
\]
and
\[
  C_0,u(j)\to C_0, u(d_0)\to  C_1, u(d_0)\to  \cdots\to
  C_n=C_{\mathrm{halt}},  u(d_0)\to  0, u(d)\to 1,d\downarrow 0,
\]
whence $u(i)\not\simeq u(j)$. This shows that $\simeq$ is not a congruence.

\medskip\noindent ($\Longleftarrow$) We argue by contraposition. Suppose that $M$
does not halt on the empty tape, and let $C_0\vdash C_1 \vdash C_2 \vdash \cdots$ be the
infinite run of $M$. We claim that $\simeq$ is the finest equivalence
relation containing the pairs
\[
  \qquad u^n(i)\simeq u^n(j)\qquad (n\in \Nat)
\]
(that is, $\simeq$ consists of these pairs, their converses, and all
diagonal pairs). This equivalence relation is clearly a
congruence. Thus let us prove the claim:
\begin{enumerate}
\item\label{T:cu:1} We have $i\simeq j$: running $i$ and $j$ on any input state $s$ yields
  \[
    s,i\to C_0, c_0\to C_1,c_0\to C_2,c_0\to \cdots
    \qquad\text{and}\qquad
    s,j\to C_0,d_0\to C_1,d_0\to C_2,d_0\to \cdots.
  \]
  Thus $i$ and $j$ are trace equivalent.

\item\label{T:cu:2} We have $u^n(i)\simeq u^n(j)$ for all $n\geq 1$: running
  $u^n(i)$ on all possible input states yields the following
  computations, where $s\in\{0,1\}$, $C\in \mathrm{Conf}$, and
  $(m,k)\in \Nat\times \{c_0,c,d_0,d\}$:
  \begin{itemize}
  \item $s, u^n(i)\to  1, c_0  \downarrow \mathsf{err}$
  \item $C, u^n(i)\to  C_0, u^n(c_0)\to  C_1, u^n(c_0)\to  C_2, u^n(c_0)\to  \cdots$
  \item $(m,k), u^n(i)\to  (m+1,k), u^{n-1}(i)\to  \cdots\to  (m+n,k), i\to  C_0, c_0\to  C_1, c_0\to  C_2, c_0\to  \cdots$
  \item $\mathsf{err}, u^n(i)\downarrow \mathsf{err}$
  \end{itemize}
  The computations for $u^n(j)$ are analogous, with $d_0$ in place
  of~$c_0$. Thus $u^n(i)$ and~$u^n(j)$ are trace equivalent.

\item We have $u^m(i)\not\simeq u^n(i)$ and $u^m(j)\not\simeq u^n(j)$
  for $m\neq n$. To see this, suppose w.l.o.g.~that $m<n$. Running $u^m(i)$
  and $u^n(i)$ on input state $(0,c_0)$ yields:
  \begin{itemize}
  \item $(0,c_0), u^m(i)\to  (1,c_0), u^{m-1}(i)\to  \cdots \to
    (m,c_0), i\to  C_0, c_0\to  C_1, c_0 \to \cdots$
  \item $(0,c_0), u^n(i)\to (1,c_0), u^{n-1}(i)\to \cdots \to (m,c_0),
    u^{n-m}(i) \to (m+1,c_0), u^{n-m-1}(i)\to \cdots$
  \end{itemize}
  Since the two computations lead to different states after $m+1$
  steps, we see that $u^m(i)\not\simeq u^n(i)$. The same argument with $d_0$ in place of $c_0$
  shows $u^m(j)\not\simeq u^n(j)$.

\item\label{T:cu:4} For each $n\in \Nat$ the program $u^n(c)$ is not
  $\simeq$-equivalent to any other program. In fact,~$u^n(c)$ is the
  unique program that on input state $(0,c)$ eventually
  terminates in the state $(n,c)$. Analogously, none of the programs
  $u^n(c_0)$, $u^n(d_0)$ and $u^n(d)$ is $\simeq$-equivalent to any
  other program. \qedhere
\end{enumerate}
\end{proof}

\begin{theorem}\label{thm:compositionality-undecidable2}
  It is undecidable whether the termination semantics induced by a
  primitive recursive \isos{} specification is compositional.
\end{theorem}

\begin{proof}
  We use the same reduction $M\mapsto \L$ as in the proof of
  \autoref{thm:compositionality-undecidable1}. For the equivalence
  relation $\approx\,=\,\approx_\L$ we show that
  \[
    \text{$M$ halts on the empty tape}
    \qquad\text{iff}\qquad
    \text{$\approx$ is not a congruence}.
  \]
  The proof of ($\Longrightarrow$) is as before (with $\simeq$ replaced by
  $\approx$). In the proof of ($\Longleftarrow$) observe that the
  equivalence relation $\approx$ is now generated by the pairs
  \[
    i\approx j \qquad\text{and}\quad u^m(i)\approx u^n(j) \qquad
    (m,n\geq 1).
  \]
  This is immediate from the analysis of computations in
  \renewcommand{\itemautorefname}{items}%
  \autoref{T:cu:1}, \ref{T:cu:2} and~\ref{T:cu:4} of the proof of
    \renewcommand{\itemautorefname}{item}%
  \autoref{thm:compositionality-undecidable1}.
\end{proof}

\subsection*{Proof of \autoref{th:sl-congruence}}

\begin{rem}
  Recall the (dual of) the classical result by
  Ad\'amek~\cite{adamek74} that for every category~$\C$ with a
  terminal object $1$ and limits of $\omega^\opp$-chains and every
  endofunctor $B$ on $\C$ preserving such limits, the final coalgebra
  $\nu B$ is obtained as the limit of the $\omega^\opp$-chain
  \[
    1\xla{!} B1 \xla{B!} B^2 1 \xla{B^2!} \cdots \xla{B^{k-1}!} B^{k}1
    \xla{B^k!} B^{k+1} 1 \xla{B^{k+1}!} \cdots
  \]
  where $!$ is the unique morphism and $B^k$ means $k$-fold
  application of $B$. For the functor $BX=\stimes{(X+1)}$ on $\Set$,
  letting $S^{\leq k}\seq S^{+}$%
  \smnote{The double superscript $S^{+,\leq k}$ looks awful; let's not
    use this; our readers will survive with $S^{\leq k}$ since we
    explained what we mean.}
  denote the set of nonempty finite
  $S$-streams of length at most $k$, we have
  $B^k 1\cong S^{\leq k} + S^k$ and $\nu B\cong
  S^{+}+S^{\omega}$. The limit projection $\pi_k\colon \nu B\to B^k1$
  maps every string $w\in S^{\mplus}$ to its prefix of length at most
  $k$ in the first coproduct component and every stream in
  $S^\omega$ to its prefix of length $k$ in the second coproduct
  component.
\end{rem}
\begin{notn}
 For $k\in \Nat$, we put
  \[
    \sbrack{-}_k
    =
    (\mu\Sigma \xto{\sbrack{-}} (\nu B)^S \xto{\pi_k^S} (B^k1)^S).
  \]
  Thus, given $p\in \mu\Sigma$ and $s\in S$, the stream
  $\sbrack{p}_k(s)\in B^k 1 = S^{\leq k}+S^k$ is the $k$-step
  behaviour of the program $p$ on the input state $s$; note that
  $\sbrack{p}_k$ retains the information whether $p$ has terminated
  after at most $k$ steps or not. For $p,q\in \mu\Sigma$  and $s \in S$,
we put
  \[
    p \simeq_k q \quad\text{iff}\quad
    \sbrack{p}_k = \sbrack{q}_k
    \qquad\text{and}\qquad
    p \simeq_{k,s} q
    \quad\text{iff}\quad
    \sbrack{p}_k(s)=\sbrack{q}_k(s).
  \]
  Observe that
  \[
    p\simeq q
    \qquad\text{iff}\qquad
    p \simeq_k q
    \quad\text{for all $k$}.
  \]
\end{notn}
\begin{proof}[Proof of \autoref{th:sl-congruence}]
  To prove that $\simeq$ is a congruence, it suffices to prove
  that each $\simeq_k$ is a congruence. To this end, we shall
  establish the following slightly stronger statement:

  \subparagraph*{Claim.}  For every $k\in \Nat$, every $n$-ary operator $\f\in \Sigma$, and $p_m,q_m\in \mu\Sigma$ for $m=1,\ldots,n$,
  \begin{enumerate}
  \item if $\f$ is passive, then
    \[
      p_m\simeq_k q_m \quad(m=1,\ldots,n)
      \qquad\implies\qquad
      \f(p_1,\ldots,p_{n}) \simeq_k \f(q_1,\ldots,q_{n}),
    \]
  \item if $\f$ is active with receiving position $j\in \{1,\ldots, n\}$ and $s\in S$, then
    \[
      p_m \simeq_k q_m
      \ (m\in \{1, \ldots, n\}\smin \{j\}) \quad\text{and}\quad
      p_j\simeq_{k,s} q_j
      \;\;\implies\;\;
      \f(p_1,\ldots,p_{n}) \simeq_{k,s} \f(q_1,\ldots,q_{n}).
    \]
  \end{enumerate}
  The proof of the claim is by induction on $k$. The induction base
  ($k=0$) holds trivially because $p\simeq_{0,s} q$ and
  $p\simeq_{0} q$ for all programs $p,q$ and states $s$.

  Now for the induction step $k\to k+1$. Further below we write $:$
  for the prefixing operation on finite lists.

  \begin{enumerate}
  \item Suppose that $\f$ is passive and that $p_m\simeq_{k+1} q_m$
    for $m=1,\ldots,n$. We need to show
    $\f(p_1,\ldots,p_{n}) \simeq_{k+1} \f(q_1,\ldots,q_{n})$, that is,
    \[
      \sbrack{\f(p_1,\ldots,p_{n})}_{k+1}(s) = \sbrack{\f(q_1,\ldots,
        q_{n})}_{k+1}(s)
      \qquad\text{for every $s\in S$}.
    \]
    Given $s\in S$, the rule applying to $s$ and $\f$ is of the form
    \[
      \inference{}{\goes{s,\f(x_{1},\dots,x_{n})}{s', t}}\qquad\text{where $t \in \Sigmas({\{x_{1},\dots,x_{n}\})}$ or $t=*$}.
    \]
    If $t=*$, then
    \[
      \sbrack{\f(p_1,\ldots,p_{n})}_{k+1}(s)
      =
      s'
      =
      \sbrack{\f(q_1,\ldots, q_{n})}_{k+1}(s).
    \]
    Otherwise $t\in \Sigmas(\{x_1,\ldots, x_{n}\})$ and
    \begin{align*}
      \sbrack{\f(p_1,\ldots, p_{n})}_{k+1}(s)
      &= s' : \sbrack{t(p_1,\ldots,p_{n})}_{k}(s') & \text{def.~of $\sbrack{-}$} \\
      &= s': \sbrack{t(q_1,\ldots,q_{n})}_k(s') \\
      &= \sbrack{\f(q_1,\ldots, q_{n})}_{k+1}(s) & \text{def.~of $\sbrack{-}$}\mathrlap{,}
    \end{align*}
    where in the second step we use that $p_m\simeq_k q_m$ for
    $m=1,\ldots,n$ and that $\simeq_k$ is a congruence by induction.

  \item Suppose that $\f$ is active with receiving position $j$. Let
    $s\in S$ and $p_m \simeq_{k+1} q_m$ for
    $m\in \{1, \ldots, n\}\smin \{j\}$ and
    $p_j\simeq_{k+1,s} q_j$. We need to show
    $\f(p_1,\ldots,p_{n}) \simeq_{k+1,s} \f(q_1,\ldots,q_{n})$, that is,
    \[
      \sbrack{\f(p_1,\ldots,p_{n})}_{k+1}(s)
      =
      \sbrack{\f(q_1,\ldots,q_{n})}_{k+1}(s).
    \]
    Let $s,p_j\to s',p$ and $s,q_j\to s', q$ be the first computation
    steps of $p_j$ and $q_j$, respectively; since $p_j\sim_{1,s} q_j$, both
    steps lead to the same state $s'$.  Suppose first that the rule applying
    to~$s$ and~$\f$ is receiving and thus of the form
    \[
      \inference{\goes{s,x_{j}}{s',y}}{\goes{s,\f(x_{1},\dots,x_{n})}{s',t}} \qquad\text{where $t=\{\f(x_{1},\dots,x_{n})[y/x_{j}]\}$ or $t=y$}.
    \]
    If $t=\f(x_{1},\dots,x_{n})[y/x_{j}]$, then
    \begin{align*}
      \sbrack{\f(p_1,\ldots,p_{n})}_{k+1}(s)
      &= s': \sbrack{\f(p_1,\ldots,p_{j-1},p,p_{j+1},\ldots, p_{n})}_{k}(s') & \text{def.~of $\sbrack{-}$} \\
      &= s': \sbrack{\f(q_1,\ldots,q_{j-1},q,q_{j+1},\ldots, q_{n})}_{k}(s') \\
      &= \sbrack{\f(q_1,\ldots,q_{n})}_{k+1}(s) & \text{def.~of $\sbrack{-}$}\mathrlap{,}
    \end{align*}
    where in the second step we use that $p_m\simeq_k q_m$ for
    $m\neq j$ and $p \simeq_{k,s'} q$, which by induction implies
    \[
      \sbrack{\f(p_1,\ldots,p_{j-1},p,p_{j+1},\ldots, p_{n})}_{k}(s')
      =
      \sbrack{\f(q_1,\ldots,q_{j-1},q,q_{j+1},\ldots, q_{n})}_{k}(s').
    \]
    If $t=y$, then
    \begin{align*}
      \sbrack{\f(p_1,\ldots,p_{n})}_{k+1}(s)
      &= s': \sbrack{p}_{k}(s') & \text{def.~of $\sbrack{-}$} \\
      &= s': \sbrack{q}_{k}(s') & \text{since $p\simeq_{k,s'} q$} \\
      &= \sbrack{\f(q_1,\ldots,q_{n})}_{k+1}(s) & \text{def.~of $\sbrack{-}$}\mathrlap{.}
    \end{align*}
    Finally, suppose that the rule applying to $s$ and $\f$ is non-receiving and thus of the form
    \[
      \inference{l_{1} && l_{2} && \cdots && l_{n}}{\goes{s,\f(x_{1},\dots,x_{n})}{s',t}}\qquad\text{where $t \in \Sigmas(\{x_{1},\dots,x_{n}\} \smin \{x_{j}\})$ or $t=*$.}
    \]
     If $t=*$, we have
    \[
      \sbrack{\f(p_1,\ldots,p_{n})}_{k+1}(s)
      =
      s'
      =
      \sbrack{\f(q_1,\ldots, q_{n})}_{k+1}(s).
    \]
Otherwise,
    \begin{align*}
      \sbrack{\f(p_1,\ldots,p_{n})}_{k+1}(s)
      &= s' : \sbrack{t(p_1,\ldots, p_{j-1},p_{j+1},\ldots, p_{n})}_k(s') & \text{def.~of $\sbrack{-}$} \\
      &= s': \sbrack{t(q_1,\ldots, q_{j-1},q_{j+1},\ldots, q_{n})}_k(s') \\
      &= \sbrack{\f(q_1,\ldots,q_{n})}_{k+1}(s) & \text{def.~of $\sbrack{-}$}\mathrlap{,}
    \end{align*}
    where in the second step we use that $p_m\simeq_k q_m$ for
    $m\neq j$ and that $\simeq_k$ is a congruence by
    induction. \qedhere
  \end{enumerate}
\end{proof}

\subsection*{Proof of \autoref{cor:str-cong}}
Let $\L$ be a streamlined \isos{} specification over the signature $\Sigma$.
  \begin{enumerate}
  \item\label{cor:str-cong:1} Suppose first that the map
    $\brack{-}=\brack{-}_\L\colon \mu\Sigma\to \nu T$ is
    surjective. Let $\f\in \Sigma$ be an $n$-ary operator and suppose that
    $p_m,q_m\in \nu T$ ($m=1,\ldots,n$) satisfy $\trc(p_m)=\trc(q_m)$ for
    all $m$. Choose $\ol{p}_m,\ol{q}_m\in \mu\Sigma$ with
    $p_m=\brack{\ol{p}_m}$ and $q_m=\brack{\ol{q}_m}$. Then we have
    for every $m$ that
    \[
      \sbrack{\ol{p}_m}
      =
      \trc(\brack{\ol{p}_m})
      =
      \trc(p_m)=\trc(q_m)
      =
      \trc(\brack{\ol{q}_m})
      =
      \sbrack{\ol{q}_m}.
    \]
    Hence,
    $\sbrack{\f(\ol{p}_1,\ldots,\ol{p}_m)} =
    \sbrack{\f(\ol{q}_1,\ldots, \ol{q}_m)}$ since the kernel $\simeq$
    of $\sbrack{-}$ is a congruence by \autoref{th:sl-congruence}. It
    follows that
    \[
      \trc(\f(p_1,\ldots, p_m))
      =
      \trc(\f( \brack{\ol{p}_1}, \ldots, \brack{\ol{p}_n}))
      =
      \trc(\brack{\f(\ol{p}_1,\ldots, \ol{p}_n)})
      =
      \sbrack{\f(\ol{p}_1,\ldots, \ol{p}_n)}
    \]
    where the second equality uses that $\sbrack{-}$ is a
    $\Sigma$-homomorphism. Similarly
    \[
      \trc(\f(q_1,\ldots, q_m)) =\sbrack{\f(\ol{q}_1,\ldots, \ol{q}_n)}.
    \]
    Thus $\trc(\f(p_1,\ldots, p_m)) = \trc(\f(q_1,\ldots, q_m))$,
    proving that the kernel of $\trc$ is a congruence.

  \item We now show that the general case can be reduced to the
    situation in \autoref{cor:str-cong:1}.  To this end, we extend the given
    specification $\L$ to a specification $\L'$ as follows:
    \begin{itemize}
    \item Extend $\Sigma$ to a signature $\Sigma'$ by adding a
      constant symbol $c_t$ for every $t\in \nu T$.

    \item Add the new rules
    \[
      \inference{}{\goes{s,c_t}{s',c_{t'}}}
      \qquad\text{for every $s\in S$ and $t\in \nu T$},
    \]
    where $s'$ is the label of the edge from the root of $t$ to its $s$-th child and $t'$ is the corresponding subtree.
  \end{itemize}
  Clearly the extended specification $\L'$ is streamlined since
  constants are passive operators, on which the streamlined format
  puts no restrictions.  Moreover $\brack{c_t}_{\L'}=t$ for every
  $t\in \nu T$, hence the map $\brack{-}_{\L'}$ is surjective. By
  \autoref{cor:str-cong:1} applied to $\L'$ we know that the kernel of
  $\trc$ is congruence w.r.t.~the $\Sigma'$-algebra structure on
  $\nu T$. In particular, it is a congruence w.r.t.~its
  $\Sigma$-algebra structure.\qed
\end{enumerate}

\subsection*{Proof of \autoref{th:fa-congruence}}
  For every pair of complex terms $p, q \in \Sigmas(\mu\Sigma)$ we
  write $(p,q) \in C[\approx]$ if $p$ and $q$ have the same shape
  (i.e.~they are equal after forgetting the labels of
  $\mu\Sigma$-labelled leaves) and matching leaves from $\mu\Sigma$
  are pairwise $\approx$-equivalent.  To prove that $\approx$ is a
  congruence, we need to show
  \[
    (p,q) \in C[\approx]
    \implies
    \hat\ini(p) \approx \hat\ini(q)
    \qquad\text{ for all $p,q\in \mu\Sigma$},
  \]
  where $\hat\ini\colon \Sigmas (\mu \Sigma) \to \mu\Sigma$ is the term
  evaluation map. By symmetry, it suffices to prove the following
  statement for all $s,\ol{s}\in S$:
  \begin{equation}\label{eq:ast}
    \text{If $(p,q) \in C[\approx]$ and $s,\hat\ini(p)$ terminates in state
      $\ol{s}$, then $s,\hat\ini(q)$ terminates in state
      $\ol{s}$.}
  \end{equation}
  The proof proceeds by induction on the number $k$ of steps until
  termination of $s,\hat\ini(p)$. Note that we count a terminating step
  $s,r\downarrow \ol{s}$ as one computation step.

  \medskip\noindent For $k=0$ the statement is vacuously true since no
  program terminates after $0$ steps.

  \medskip\noindent For the induction step, let $k>0$ and suppose that
  $s,\hat\ini(p)$ terminates after $k$ steps in the state $\ol{s}$. We need
  to show that also $s,\hat\ini(q)$ terminates in $\ol{s}$. This is shown
  by induction on the structure of $p$:
  \begin{itemize}
  \item If $p$ is a leaf of the term tree, i.e.~an element of
    $\mu\Sigma$, then $(p,q)\in C[\approx]$ implies that~$q$ has the
    same form and $p\approx q$. It follows that
    $\hat\ini(p)\approx \hat\ini(q)$, since $\hat\ini(p)=p$ and
    $\hat\ini(q)=q$.

  \item Now suppose that $p = \f(p_{1},\dots,p_{n})$ and
    $q = \f(q_{1},\dots,q_{n})$ for some $n$-ary operator $\f$,
    $n>0$. From $(p,q)\in C[\approx]$, it follows that
    $(p_i,q_i)\in C[\approx]$ for $i=1,\dots,n$. We distinguish two
    cases:

    \begin{enumerate}
    \item The operator $\f$ is passive. Then $s,p$ and $s,q$ trigger the same rule
      \[
        \inference{}{\goes{s,\f(x_{1},\dots,x_{n})}{s', t}}\qquad\text{where $t\in \Sigmas({\{x_{1},\dots,x_{n}\})}$ or $t=*$}.
      \]
      If the rule is terminating ($t=*$) then both $\hat\ini(p)$ and~$\hat\ini(q)$
      terminate after one step in state $\ol{s}=s'$. Otherwise,
      their respective first computation steps are
      \[
        s,\hat\ini(p)\to s',\hat\ini(p')
        \qquad\text{and}\qquad
        s,\hat\ini(q)\to s',\hat\ini(q')
      \]
      for $p',q'\in \Sigmas(\mu\Sigma)$ given by
      $p'=t[p_1/x_1,\ldots,p_{n}/x_{n}]$ and
      $q'=t[q_1/x_1,\ldots, q_{n}/x_n]$. Since
      $(p_i,q_i)\in C[\approx]$ for $i=1,\dots,n$, we have
      $(p',q')\in C[\approx]$. Moreover, $s',\hat\ini(p')$ terminates
      in the state $\ol{s}$ in strictly less than $k$ steps, so by
      induction, $s',\hat\ini(q')$ terminates in the state
      $\ol{s}$. It follows that also $s,\hat\ini(q)$ terminates in the state $\ol{s}$.

    \item The operator $\f$ is active with receiving position
      $j$. Then $s,\hat\ini(p)$ and $s,\hat\ini(q)$ will have to
      execute their subterms $\hat\ini(p_j)$ and $\hat\ini(q_j)$,
      respectively.  Executing $s,\hat\ini(p_j)$ requires at most~$k$
      steps and terminates in some state $s'$. Since
      $(p_j,q_j)\in C[\approx]$, the inner induction hypothesis
      implies that $s,\hat\ini(q_j)$ also terminates in $s'$. Upon
      termination of~$\hat\ini(p_j)$ or~$\hat\ini(q_j)$, respectively, both
      $s,\hat\ini(p)$ and $s,\hat\ini(q)$ therefore trigger a rule
      \[
        \inference{\rets{?,x_{j}}{s'}}{\goes{?,\f(x_{1},\dots,x_{n})}{s'',t}}\qquad\text{      where
      $t \in \Sigmas(\{x_{1},\dots,x_{n}\} \smin \{x_{j}\})$ or $t=*$.}
      \]
      Note that while the input $?$ might be different for the two
      computations, the output $s''$ and target $t$, respectively, are the same for both
      computations because they only depend on $s'$.

      If the above rule is terminating ($t=*$), then $s,\hat\ini(p)$ and
      $s,\hat\ini(q)$ both terminate in the state $\ol{s}=s''$, and we
      are done. Otherwise, $s,\hat\ini(p)$ evolves to
      $s'',\hat\ini(p')$, and $s,\hat\ini(q)$ evolves to
      $s'',\hat\ini(q')$ in finitely many steps, where
      $p',q'\in \Sigmas(\mu\Sigma)$ are given by
      $p'=t[p_m/x_m : m\neq j]$ and $q'=t[q_m/x_m : m\neq j]$;
      moreover, $s'',\hat\ini(p')$ terminates in state $\ol{s}$ in
      strictly less than~$k$ steps. Since $(p_i,q_i)\in C[\approx]$
      for all~$i$, we have $(p',q')\in C[\approx]$. By the outer
      induction hypothesis, we conclude that $s'',\hat\ini(q')$
      terminates in $\ol{s}$. Consequently, also $s,\hat\ini(q)$
      terminates in $\ol{s}$.\qed
    \end{enumerate}
\end{itemize}

\subsection*{Proof of \autoref{cor:fns-cong}}
The proof is completely analogous to the one of
\autoref{cor:str-cong}; just note that extending a cool \isos{}
specification by constants preserves coolness. \qed

\end{document}
